\documentclass[12pt]{article}
\usepackage{amsmath,amssymb}
\usepackage[utf8]{inputenc}
\usepackage{xspace}
\usepackage[dvipsnames,usenames]{xcolor}
\usepackage{tikz}
\usetikzlibrary{shapes}

\tikzset{
  avertex/.style={circle,draw,inner sep=2pt,fill=acolor},
  bvertex/.style={circle,draw,inner sep=2pt,fill=bcolor},
  cvertex/.style={circle,draw,inner sep=2pt,fill=ccolor},
  avertex/.style={regular polygon,regular polygon sides=4,draw,inner sep=2pt,fill=white},
  bvertex/.style={regular polygon,regular polygon sides=4,draw,inner sep=2pt,fill=Gray},
  cvertex/.style={regular polygon,regular polygon sides=4,draw,inner sep=2pt,fill=Black},
  vertex/.style={circle,draw,inner sep=2pt,fill=white},
  vvertex/.style={circle,draw,inner sep=2pt,fill=Gray},
  vvvertex/.style={circle,draw,inner sep=2pt,fill=Black},
  bbvertex/.style={regular polygon,regular polygon sides=3,draw,inner sep=1.6pt,fill=Gray}
}

\def\afterthmseparator{.}
\makeatletter
\renewcommand{\@begintheorem}[2]{\trivlist
      \item[\hskip \labelsep{\bf #1\ #2\unskip\afterthmseparator}]\em}
\renewcommand{\@opargbegintheorem}[3]{\trivlist
      \item[\hskip \labelsep{\bf #1\ #2\ (#3)\unskip\afterthmseparator}]\em}
\makeatother
\newtheorem{theorem}{Theorem}[section]
\newtheorem{lemma}[theorem]{Lemma}
\newtheorem{corollary}[theorem]{Corollary}
\newtheorem{remark}[theorem]{Remark}
\newcommand{\bull}{\mbox{$\;\;\;$\vrule height .9ex width .8ex depth -.1ex}}
\newenvironment{proof}{\par\smallbreak\noindent{\bf Proof.~}}
{\unskip\nobreak\hfill \bull \par\medbreak}
\newcommand{\noproof}{~\bull}
\newcounter{oq}
\newcommand{\que}{\refstepcounter{oq}\par{\bf \theoq.}~}

\newcommand{\hide}[1]{}
\newcommand{\refeq}[1]{(\ref{eq:#1})}
\newcommand{\setdef}[2]{\left\{ \hspace{0.5mm} #1 : \hspace{0.5mm} #2 \right\}}
\newcommand{\msetdef}[2]{\left\{\!\!\left\{ \hspace{0.5mm} #1 : \hspace{0.5mm} #2 \right\}\!\!\right\}}
\newcommand{\function}[2]{:#1 \rightarrow #2}

\newcommand{\cclass}[1]{\textsf{\upshape #1}}
\newcommand{\ac}[1]{\cclass{AC$^{\cclass{#1}}$}\xspace}

\newcommand{\p}{\cclass{P}\xspace}

\newcommand{\stabi}[1]{\mathit{Stab}(#1)}
\newcommand{\fo}[1]{\ensuremath{\mathrm{FO}^{#1}}\xspace}
\newcommand{\twoclogic}{\ensuremath{\mathrm{FO}^2_\#}\xspace}
\newcommand{\kclogic}{\ensuremath{\mathrm{FO}^k_\#}\xspace}
\newcommand{\game}{\mbox{\sc Game}}
\newcommand{\vtype}[1]{%
\raisebox{.2mm}{%
\begin{tikzpicture}
\node[#1] at (0,0) {};
\end{tikzpicture}\,}}
\newcommand{\barG}{\overline{G}}
\newcommand{\barH}{\overline{H}}

\begin{document}

\title{Universal covers, color refinement, and\\ two-variable counting logic:\\
Lower bounds for the depth}

\author{Andreas Krebs\thanks{Wilhelm-Schickard-Institut, Universit\"at T\"ubingen, Sand 13, 72076 T\"ubingen, Germany.}\ \
and Oleg Verbitsky\thanks{%
Humboldt-Universit\"at zu Berlin,
Institut f\"ur Informatik,
Unter den Linden 6,
D-10099 Berlin.
Supported by DFG grant VE 652/1--1.
On leave from the Institute for Applied Problems of Mechanics and Mathematics,
Lviv, Ukraine.}}

\date{}

\maketitle

\begin{abstract}
Given a connected graph $G$ and its vertex $x$, let $U_x(G)$ denote the universal cover of $G$
obtained by unfolding $G$ into a tree starting from $x$. 
Let $T=T(n)$ be the minimum number such that,
for graphs $G$ and $H$ with at most $n$ vertices each,
the isomorphism of $U_x(G)$ and $U_y(H)$ surely follows from the isomorphism
of these rooted trees truncated at depth $T$.
Motivated by applications in theory of distributed computing,
Norris [Discrete Appl.\ Math.\ 1995] asks if $T(n)\le n$.
We answer this question
in the negative by establishing that $T(n)=(2-o(1))n$.
Our solution uses basic tools of finite model theory
such as a bisimulation version of the Immerman-Lander 2-pebble counting game.

The graphs $G_n$ and $H_n$ we construct to prove the lower bound for $T(n)$
also show some other tight lower bounds.
Both having $n$ vertices, $G_n$ and $H_n$ can be distinguished
in 2-variable counting logic only with quantifier depth $(1-o(1))n$.
It follows that \emph{color refinement}, the classical procedure used
in isomorphism testing and other areas for computing the coarsest
equitable partition of a graph, needs $(1-o(1))n$ rounds to achieve
color stabilization on each of $G_n$ and $H_n$. Somewhat surprisingly,
this number of rounds is not enough for
color stabilization on the disjoint union of $G_n$ and $H_n$,
where $(2-o(1))n$ rounds are needed.
\end{abstract}

\section{Introduction}

A homomorphism from a connected graph $H$ onto a graph $G$
is called a \emph{covering map} if it is a bijection
in the neighborhood of each vertex of $H$.
In this case, we say that $H$ is a \emph{cover} of $G$ or that $H$ \emph{covers} $G$. 
Given a vertex $x$ of $G$, let $U_x(G)$ 
denote the unfolding of $G$ into a tree starting from $x$.
This tree is called the \emph{universal cover} of $G$
because $U_x(G)$ covers every cover $H$ of $G$; see examples in Fig.~\ref{fig:simple-example}.

These notions appeared in algebraic and topological graph theory \cite{Biggs,CvetkovicDS,Massey,Reidemeister},
where they are helpful, for instance, in factorization
of the characteristic polynomial of a graph \cite{Sachs64}
or in classification of projective planar graphs \cite{Negami88}.
Further applications were found in such diverse areas as
finite automata theory \cite{Moore56,Norris95},
combinatorial group theory  \cite{Stallings83,Stillwell},
finite model theory \cite{Otto13},
construction of expander graphs \cite{AmitL06,BiluL06}
and, maybe most noticeably, in
distributed computing \cite{Angluin80,Bodlaender89,FischerLM86,YamashitaK88}
(see also the surveys \cite{Kranakis97,LitovskyMS99}).
The problem we consider in this paper arose in the last area.

In theory of distributed systems, a synchronous network of anonymous processors
is presented as a graph $G$ where, in a unit of time, two processors exchange
messages if the corresponding vertices are connected by an edge.
The processors are supposed to have unlimited computational power (for example,
being automata with unbounded number of states).
The processors at vertices of the same degree execute the same program.
They are identical at the beginning, but later can have different states if
they receive different messages from their neighbors. 
All processors have a common goal of arriving at a specified configuration of their states.
For example, the \emph{leader election} problem is to ensure that exactly one of the
processors comes in the distinguished state ``elected'' while all others
come in the ``unelected'' state.

For an integer $t\ge0$, let $U^t_x(G)$ be the truncation of $U_x(G)$ at depth $t$. 
Two processors $x$ and $y$ are indistinguishable by their states up until time $t$
if $U^t_x(G)\cong U^t_y(G)$, where $\cong$ denotes isomorphism of rooted trees.
In particular, none of $x$ and $y$ can be elected in time $t$ as a leader.
Moreover, leader election is possible only if $G$ contains a vertex $x$
such that $U_x(G)\not\cong U_y(H)$ for any other $y$ (see, e.g., the discussion in~\cite{BoldiV02}).

Another archetypical problem in distributed computing is \emph{network topology recognition},
that is, identification of the isomorphism type of the underlying graph $G$ or,
at least, checking if $G$ has a specified property $\mathcal P$.
Angluin \cite{Angluin80} observed that, if $H$ covers $G$, then
these two graphs are indistinguishable in the above model of distributed computation
(and in other natural models of local computations). It follows that
a class of graphs $\mathcal P$ is recognizable only if it is closed under coverings.  
This basic observation is used, for example, in \cite{CourcelleM94} 
where it is shown that, except for a few special cases,
minor-closed classes of graphs are not closed under covers and, hence,
cannot be recognized.

Two networks $G$ and $H$ cannot be distinguished in time $t$ by the states
of their processors $x\in V(G)$ and $y\in V(H)$ if $U^t_x(G)\cong U^t_y(H)$.
Suppose that each of $G$ and $H$ has at most $n$ nodes. What is the minimum time
$T=T(n)$ such that $U^T_x(G)\cong U^T_y(H)$ surely implies $U_x(G)\cong U_y(H)$?
In other words, $T(n)$ is the minimum time that suffices to distinguish $x$ and $y$ 
whenever at all possible. 

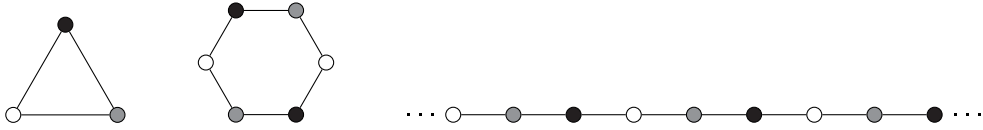
\begin{figure}
\centering
 \begin{tikzpicture}[scale=-.8]
 \path (30:1cm) node[vertex] (a)   {}
                 (150:1cm) node[vvertex] (b)   {} edge (a)
                 (270:1cm) node[vvvertex] (c)   {} edge (a) edge (b);
 \end{tikzpicture}
\qquad
 \begin{tikzpicture}[scale=.8]
 \path (0:1cm) node[vertex] (a)   {}
                 (60:1cm) node[vvertex] (b)   {} edge (a)
                 (120:1cm) node[vvvertex] (c)   {} edge (b)
                 (180:1cm) node[vertex] (d)   {} edge (c)
                 (240:1cm) node[vvertex] (e)   {} edge (d)
                 (300:1cm) node[vvvertex] (f)   {} edge (a) edge (e)
;
 \end{tikzpicture}
\qquad
 \begin{tikzpicture}[scale=.8]
  \path 
                 (-.75,0) node[inner sep=.5pt,fill=black] () {}
                 (-.55,0) node[inner sep=.5pt,fill=black] () {}
                 (-.35,0) node[inner sep=.5pt,fill=black] () {}                
                   (0,0) node[vertex] (a0)   {}
                 ++(1,0) node[vvertex] (b0)   {} edge (a0)
                 ++(1,0) node[vvvertex] (c0)   {} edge (b0)
                 ++(1,0) node[vertex] (a1) {} edge (c0)
                 ++(1,0) node[vvertex] (b1)   {} edge (a1)
                 ++(1,0) node[vvvertex] (c1)   {} edge (b1)
                 ++(1,0) node[vertex] (a2) {} edge (c1)
                 ++(1,0) node[vvertex] (b2)   {} edge (a2)
                 ++(1,0) node[vvvertex] (c2)   {} edge (b2)
                ++(.35,0) node[inner sep=.5pt,fill=black] () {}
                 ++(.2,0) node[inner sep=.5pt,fill=black] () {}
                 ++(.2,0) node[inner sep=.5pt,fill=black] () {}
;
 \end{tikzpicture}
\caption{The cycle $C_6$ covers the cycle $C_3$. The infinite path is
the universal cover of both cycles. The corresponding covering maps
are visualized by colors of vertices.}
\label{fig:simple-example}
\end{figure}

Norris \cite{Norris95} considers the case when $G=H$ and shows that then $T(n)\le n-1$,
improving upon an earlier bound of $n^2$ in \cite{YamashitaK88}
and thereby increasing the performance of several distributed algorithms.
In the same paper, she asks if $T(n)\le n$ in the general case of two graphs $G$ and $H$.
We answer this question in the negative by establishing that $T(n)=(2-o(1))n$.
The upper bound $T(n) < 2n$ is standard, and our main contribution is a construction
of graphs $G$ and $H$
showing a lower bound $T(n)\ge 2n-16\sqrt n$. 

The bound of $2n$ is a standard upper bound for the communication
round complexity of the distributed algorithms that are based on
computing the isomorphism type of the universal cover $U_x(G)$
or on related concepts; see Tani \cite{Tani12}.
Our result implies that this bound is tight up to a term of $o(n)$
for any algorithm that aims at gaining all knowledge about the
network $G$ available to a particular party~$x$.

Our solution of Norris's problem uses, perhaps for the first time
in the area of distributed computing, a conceptual framework
that was created in finite model theory. Specifically,
our proof of the lower bound for $T(n)$ makes use of
a bisimulation version of the 2-pebble counting game \cite{ImmermanL90},
that was used for diverse purposes in \cite{GraedelO99,AtseriasM13}.
The interplay between the two areas is discussed in more detail below.

Note that Norris in \cite{Norris95} considers directed graphs and
allows multiple edges and loops.
This setting is more general and as well important for modelling of distributed systems:
Undirected graphs we consider in this paper 
just correspond to networks with bidirectional communication channels.
It should be stressed that our lower bound for $T(n)$, while shown for undirected graphs,
holds true in the setting of \cite{Norris95} by considering a simple orientation of
the constructed graphs; see Remark~\ref{rem:orgraphs}.

\smallskip

\noindent
{\bf Relation to color refinement.}
A coloring of the vertex set of an undirected graph $G$ is called \emph{stable} 
if any two equally colored vertices have the same number of neighbors of each color.
In algebraic graph theory \cite{GodsilR}, the corresponding partition of $V(G)$ is called \emph{equitable}. 
The well-known \emph{color refinement} procedure begins with a uniform coloring of $V(G)$
and refines it step by step so that, if two vertices currently have equal colors but
differently colored neighborhoods, then their new colors are different.
The procedure terminates as soon as no further refinement is possible.
The coloring obtained in this way is stable and gives a unique coarsest equitable partition of $V(G)$.
More efficient implementations of color refinement have been developed in the literature
starting from Hopcroft's work on minimization of finite automata \cite{Hopcroft71}; see \cite{BerkholzBG13}
and references therein.
Even in the simplest version, color refinement is used in isomorphism testing as 
a very efficient way to compute a canonical labeling for almost all input graphs 
(Babai, Erd{\"o}s, and Selkow \cite{BabaiES80}).
Angluin \cite{Angluin80}
uses color refinement to decide if two given graphs have a common cover.
The last condition is important because, as follows from the discussion above, 
such graphs are indistinguishable by local computations.

Using the relationship between universal covers and color refinement observed in \cite{Angluin80},
our bound for $T(n)$ can be recasted as a result about the basic color refinement procedure. 
Let $\stabi G$ denote the number of refinements
made on the input $G$ till stabilization. 
The graphs $G_n$ and $H_n$ we construct to solve Norris's problem show that
the obvious upper bound $\stabi G<n$, where $n$ is the number of vertices in $G$,
is tight. Both $G_n$ and $H_n$  have $n$ vertices,
and both $\stabi{G_n}$ and $\stabi{H_n}$ are $(1-o(1))n$.
Moreover, the disjoint union $G_n\cup H_n$ of these graphs demonstrates
a counterintuitive phenomenon: $\stabi{G_n\cup H_n}=(2-o(1))n$, which means that
stabilization of the disjoint union can require much more refinement rounds
than stabilization of each component.

\smallskip

\noindent
{\bf Relation to two-variable counting logic.}
As it is already mentioned,
our main technical tool is a version of the 2-pebble counting game
introduced by Immerman and Lander \cite{ImmermanL90}
for analysis of the expressibility of \twoclogic,
first-order logic with two variables and counting quantifiers. Let $D(G)$ denote
the minimum quantifier depth of a formula defining the equivalence type
of a graph $G$ in \twoclogic. 
The connection to color refinement established by Immerman and Lander in \cite{ImmermanL90}
can be used to show that $D(G)\le \stabi G+2\le n+1$ for $G$ with $n$ vertices.
Our construction yields graphs asymptotically attaining this bound. 
While a linear lower bound of $D(G)\ge n/2-O(1)$
can be easily seen by considering $G=P_n$, the path graph on $n$ vertices, the exact
asymptotics of the maximum value of $D(G)$ is determined for the first time. Note also that
the parameter $D(G)$ is very small in the average case. 
As follows from \cite{BabaiES80}, $D(G)\le \stabi G+2\le 4$ for almost all $G$ on $n$ vertices.

\smallskip

\noindent
{\bf Related work.}
A \emph{fibration} is an analog of a covering for directed graphs, possibly with loops
and multiple edges, where
the local bijectivity property is required only on \emph{in}-arrows.
Boldi and Vigna \cite{BoldiV02} discuss the Norris's question for fibrations and 
notice that the lower bound $T(n)\ge 2n-2$ can in this setting be obtained
by considering the following simple example of digraphs $G$ and $H$ 
along with vertices $x\in V(G)$ and $y\in V(H)$.
Both $G$ and $H$ are obtained from
the bidirectional path on the vertices $1,2,\ldots,n$ by cloning the arrow from $n-1$ to $n$
in $G$ and by adding a loop at $n$ in $H$; furthermore, $x=1$ and $y=1$.
The overall idea behind our construction is actually similar, though its
implementation without loops and multiple edges 
(and with the bijectivity constraint on the entire neighborhood of a vertex)
is not so simple.

In a more realistic model of a bidirectional network, 
an undirected graph $G$ is endowed with a \emph{port-numbering}. 
This is a labeling that determines an order on
the set of the incident edges for each node. Hendrickx \cite{Hendrickx13} proves in this setting%
\footnote{%
In fact, in place of the universal covers Hendrickx \cite{Hendrickx13} considers
the related concept of a \emph{view} introduced by Yamashita and Kameda
in \cite{YamashitaK88}. While the vertices of the universal cover $U_x(G)$
can be identified with the non-backtracking walks in $G$ starting at $x$,
the vertices of the view of $x\in V(G)$ correspond to all (not necessarily non-backtracking)
walks in $G$ starting at $x$. Isomorphism of truncated views is equivalent to
isomorphism of truncated universal covers because both can be characterized
by color refinement as in Lemma~\ref{lem:UvsC}.
}
that $U^t_x(G)\cong U^t_y(G)$ implies $U_x(G)\cong U_y(G)$ for $t=O(d+d\log(n/d))$
where $d$ denotes the diameter of $G$. This bound can be preferable to Norris's bound
of $n-1$ if $d=o(n)$. The optimality of this bound is shown by Dereniowski, Kosowski, 
and Pajak~\cite{DereniowskiKP14}.
Note that our solution of Norris's problem can be extended also to
port-numbered graphs; see Remark~\ref{rem:port-num}.

The results on the computational complexity of deciding if $H$ covers $G$ for two given graphs
are surveyed in~\cite{FialaK08}.

\smallskip

\noindent
{\bf Organization of the paper.}
The connection between universal covers and color refinement 
is explored in Section \ref{s:refinement}.
A key technical role is here played by Lemma \ref{lem:trees},
which is a kind of a reconstructibility result for
rooted trees (cf.\ \cite{Kelly57,Nesetril71}).
Once this connection is established (Lemma \ref{lem:UvsC}), it 
readily yields the upper bound $T(n) < 2n$ (Lemma \ref{lem:twon}).
The lower bound for $T(n)$ is obtained in Section \ref{s:Norris} (Theorem \ref{thm:nancy}).
The proof uses the relationship between
color refinement and the bisimulation version of
the 2-pebble Immerman-Lander game \cite{ImmermanL90}.
The core of the proof is a construction of $n$-vertex graphs $G$ and $H$
containing vertices $u$ and $v$ respectively such that,
while Spoiler can win starting from the position $(u,v)$,
Duplicator can resist during $(2-o(1))n$ rounds.
In Section \ref{s:logic} we analyse the original version
of the Immerman-Lander game on the same graphs and determine
the maximum quantifier depth needed to define
the equivalence type of a graph on $n$ vertices in \twoclogic (Theorem \ref{thm:qdepth}).
Theorem \ref{Phard} in the same section is obtained as a by-product
and stated for expository purposes. It shows another connection between
universal covers and logic: 
Deciding  if given graphs have a common cover is \p-hard by a reduction
from the \twoclogic-equivalence problem, whose \p-hardness
is established by Grohe \cite{Grohe99}.
This reduction is implicitly contained in \cite[Theorem 2.2]{RamanaSU94}; see also~\cite{AtseriasM13}.
We conclude with some open questions in Section~\ref{s:open}.

\section{Universal covers and color refinement}\label{s:refinement}

\subsection{Basic definitions and facts}\label{s:defs}

Unless stated otherwise, we consider unlabeled undirected graphs without loops
and multiple edges.
Given a graph $G$, we denote its vertex set by $V(G)$. 
The \emph{neighborhood} of a vertex $v$ consists of
all vertices adjacent to $v$ and is denoted by~$N(v)$.

In this section we consider, along with finite, also infinite graphs.
The following definitions, which are srandard in the finite case,
apply as well to infinite graphs.
A graph is \emph{connected} if any two its vertices can
be connected by a (finite) path.
A graph is \emph{acyclic} if it contains no (finite) cycle.
A (possibly infinite) \emph{tree} is an acyclic connected graph.

All graphs in this section are supposed to be \textbf{connected}.
Let $\alpha$ be a homomorphism from $H$ \emph{onto} $G$.
If $\alpha$ is a bijection from $N(v)$ onto $N(\alpha(v))$
for each $v\in V(H)$, then it is called a \emph{covering map},
and $H$ is called a \emph{covering graph} (or a \emph{cover}) of $G$.
Sometimes we also say that $H$ \emph{covers} $G$.
Note that the following fact holds true both for finite and infinite graphs.

\begin{lemma}\label{lem:iso-over-trees}
If $\alpha$ is a covering map from a connected graph $H$ onto
a tree $T$, then $H$ is a tree and $\alpha$ is an isomorphism from
$H$ to~$T$.  
\end{lemma}

\begin{proof}
If $H$ had a cycle, $\alpha$ would take it to a closed non-backtracking walk in $T$,
which is impossible because such a walk must contain a cycle. 
Next we show that $\alpha$ is injective.
Fix a vertex $x$ in $H$ and let $\alpha^i$ denote the restriction of $\alpha$
to the vertices of $H$ at the distance at most $i$ from $x$. 
For every $i\ge0$, $\alpha^i$ is injective; this follows by
induction on $i$. Therefore, $\alpha$ is injective. Finally, note that 
an injective covering map is an isomorphism because it always takes 
a pair of non-adjacent vertices to a pair of non-adjacent vertices.
\end{proof}

Given a connected graph $G$ and a vertex $x\in V(G)$, define a graph $U_x(G)$ as follows.
The vertex set of $U_x(G)$ consists of non-backtracking walks in $G$
starting at $x$, that is, of sequences $(x_0,x_1,\ldots,x_k)$ such that
$x_0=x$, $x_i$ and $x_{i+1}$ are adjacent, and $x_{i+1}\ne x_{i-1}$.
Two such walks are adjacent in $U_x(G)$ if one of them extends the other
by one component, that is, one is $(x_0,\ldots,x_k,x_{k+1})$
and the other is $(x_0,\ldots,x_k)$. 
Notice the following properties
of this construction.

\begin{lemma}\label{lem:UxG}\mbox{}
 \begin{enumerate}
\item[\bf 1.]
$U_x(G)$ is a tree.
\item[\bf 2.] 
The map $\gamma_G$ defined by $\gamma_G(x_0,\ldots,x_k)=x_k$
is a covering map from $U_x(G)$ to~$G$.
\item[\bf 3.]
If $\alpha$ is a covering map from $H$ onto $G$, then
there is a covering map $\beta$ from $U_x(G)$ onto $H$
such that $\gamma_G=\alpha\circ\beta$.
\end{enumerate} 
\end{lemma}

\begin{proof}
We skip Parts 1 and 2 that can be shown by a direct argument.
To prove Part 3, fix a vertex $y$ of $H$ such that $\alpha(y)=x$.
Define $\bar\alpha(y_0,\ldots,y_k)=(\alpha(y_0),\ldots,\alpha(y_k))$
and note that $\bar\alpha$ is a covering map from $U_y(H)$ to $U_x(G)$.
Since the two graphs are trees, $\bar\alpha$ is an isomorphism by Lemma \ref{lem:iso-over-trees}.
Note that $\gamma_G\circ\bar\alpha=\alpha\circ\gamma_H$.
Thus, we can set $\beta=\gamma_H\circ\bar\alpha^{-1}$.  
\end{proof}

Call $U$ a \emph{universal cover} of $G$ if $U$ covers
any covering graph of $G$. Lemma \ref{lem:UxG}.3 implies that $U_x(G)$
is a universal cover of $G$.
The next lemma shows that we could define the universal cover of $G$,
uniquely up to isomorphism, as a tree covering~$G$.

\begin{lemma}\label{lem:univ-cover}\mbox{}
\begin{enumerate}
\item[\bf 1.] 
All universal covers of $G$ are isomorphic trees.
\item[\bf 2.] 
Any tree covering $G$ is a universal cover of $G$.
\end{enumerate}
\end{lemma}

\begin{proof}
  {\bf 1.}
By definition, all universal covers of $G$ cover each other.
Since one of them, namely $U_x(G)$, is a tree, the claim follows
from Lemma~\ref{lem:iso-over-trees}.

{\bf 2.}
Assume that a tree $T$ covers $G$. Let $U$ be a universal cover of $G$.
Then $U$ covers $T$, and the two graphs are isomorphic by Lemma \ref{lem:iso-over-trees}.
Like $U$, the tree $T$ is therefore a universal cover of~$G$. 
\end{proof}

\begin{lemma}\label{lem:UGUH}\mbox{}
  \begin{enumerate}
  \item[\bf 1.] 
  If $C$ covers $G$, then their universal covers $U_C$ and $U_G$ are isomorphic.
\item[\bf 2.] 
If $G$ and $H$ have a common covering graph, then their universal covers $U_G$ and $U_H$ are isomorphic.
  \end{enumerate}
\end{lemma}

\begin{proof}
{\bf 1.}
By Lemma \ref{lem:univ-cover}.1, $U_C$ is a tree.
Since $U_C$ covers $C$ and $C$ covers $G$, $U_C$ covers $G$.
By Lemma \ref{lem:univ-cover}.2, $U_C$ is a universal cover of $G$
and, therefore, is isomorphic to~$U_G$.

{\bf 2.} 
Consider a common cover $C$ of $G$ and $H$ and apply Part~1.
\end{proof}

Lemma \ref{lem:UGUH} shows that two 
connected graphs have a common covering graph if and only if their universal covers are isomorphic.
It is known (Leighton \cite{Leighton82})
that graphs have a common covering graph if and only if they
have a common \emph{finite} covering graph.


\subsection{A reconstruction lemma for rooted trees}

A \emph{rooted tree $T_v$} is a tree with one distinguished vertex $v$, which is
called {\em root}. An isomorphism of rooted trees
should not only preserve the adjacency relation but also map the root
to the root. The {\em depth\/} of $T_v$ is the maximum distance from $v$ to a leaf.
We write $T^i_v$ to denote the subgraph of $T_v$ induced on the vertices
at the distance at most $i$ from $v$. Note that $T^i_v$ inherits the root~$v$.

If $v\in V(G)$, let $G-v$ denote the graph obtained by removing $v$ from $G$ along with all incident edges.
If $w\in N(v)$, then $T_v(w)$ denotes the component of $T_v-v$ 
containing $w$.
The tree $T_v(w)$ will be supposed to be rooted at $w$
and will be referred to as a \emph{branch} of~$T_v$.

There are well-known results about the reconstructibility of the isomorphism type of a tree
from the isomorphism types of some of its proper subgraphs \cite{Kelly57,Nesetril71}. 
We will need a result of this kind showing that a rooted tree is
reconstructible, up to isomorphism, from a family of related rooted trees
of smaller depth.

\begin{lemma}\label{lem:trees}
  Let $T$ and $S$ be 
trees and $x\in V(T)$ and $y\in V(S)$ be their vertices of the same degree with neighborhoods
$N(x)=\{x_1,\ldots,x_k\}$ and $N(y)=\{y_1,\ldots,y_k\}$. Let $r\ge1$.
Suppose that $T^{r-1}_x\cong S^{r-1}_y$ and $T^r_{x_i}\cong S^r_{y_i}$ for all $i\le k$.
Then $T^{r+1}_x\cong S^{r+1}_y$.
\end{lemma}

\begin{proof}
Note that, for each $i\le k$, $x_i$ and $y_i$ have equal degrees.
Let $A_0,A_1,\ldots,A_m$ be the branches of $T^r_{x_i}$
(that is, the components of $T^r_{x_i}-x_i$ rooted at the neighbors of $x_i$).
The dependence on $i$ in the notation $A_j$ is dropped.
Similarly, let $B_0,B_1,\ldots,B_m$ be the branches of $S^r_{y_i}$.
We suppose that $A_0=T^r_{x_i}(x)$ and $B_0=S^r_{y_i}(y)$, that is,
the roots of $A_0$ and $B_0$ are $x$ and $y$ respectively.

Our task is to prove that $A_0\cong B_0$ for each $i\le k$.
Along with the conditions $T^r_{x_i}\cong S^r_{y_i}$,
this will imply that $T^{r+1}_x(x_i)\cong S^{r+1}_y(y_i)$ for all $i\le k$.
Merging these isomorphisms with 
the map taking $x$ to $y$ and each $x_i$ to $y_i$, 
we will obtain a desired isomorphism from $T^{r+1}_x$ to $S^{r+1}_y$.

Consider an isomorphism $\alpha_i$ from $T^r_{x_i}$ to $S^r_{y_i}$.
If $\alpha_i$ maps $A_0$ onto $B_0$, we are done. Otherwise,
without loss of generality we can assume that $\alpha_i$
maps $A_0$ onto $B_1$, $A_1$ onto $B_0$, and $A_j$ onto $B_j$
for every $2\le j\le m$.

Since both $A_0$ and $B_0$ have depth at most $r-1$, we need to show that
$A_0^{r-1}\cong B_0^{r-1}$, where $A_0^{h}=(A_0)_x^{h}$ and similar notation is used
also for the rooted tree $B_0$. Using induction, we will prove that
$A_0^{h}\cong B_0^{h}$ for all $h=0,1,\ldots,r-1$.

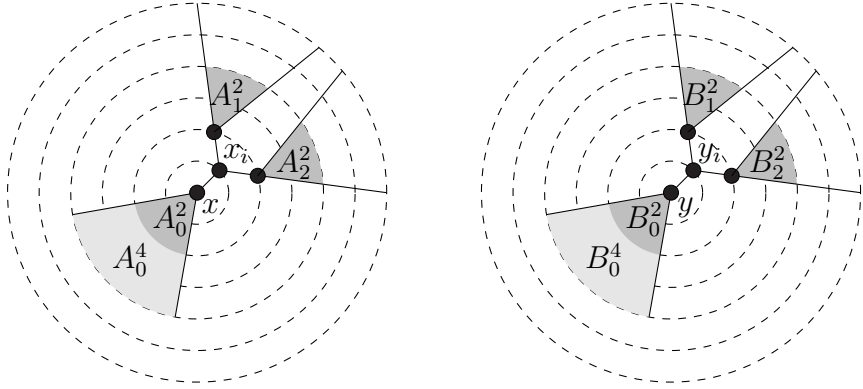
\begin{figure*}
  \centering
\begin{tikzpicture}[every node/.style={circle,draw,inner sep=2pt,fill=Black},scale=.84]
  \begin{scope}[scale=.5]
\draw[dashed] (0,0) circle (1cm) circle (2cm) circle (3cm) circle (4cm) circle (5cm) circle (6cm);
\coordinate (a1) at (.535,1.92);
\coordinate (a2) at (1.92,.535);
\begin{scope}
  \clip (0,0) circle (4cm);
\clip (190:6cm) -- (0,0) -- (260:6cm);
\path[fill=gray!20] (190:6cm) -- (0,0) -- (260:6cm);
\end{scope}
\begin{scope}
  \clip (0,0) circle (2cm);
\clip (190:4cm) -- (0,0) -- (260:4cm);
\path[fill=gray!50] (190:4cm) -- (0,0) -- (260:4cm);
\end{scope}
\begin{scope}
  \clip (0,0) circle (4cm);
\clip (a1) -- (0,6) -- (50:6cm);
\path[fill=gray!50] (a1) -- (0,6) -- (50:6cm);
\end{scope}
\begin{scope}
  \clip (0,0) circle (4cm);
\clip (a2) -- (6,0) -- (40:6cm);
\path[fill=gray!50] (a2) -- (6,0) -- (40:6cm);
\end{scope}

\draw 
(0,0) -- (190:4cm) (0,0) -- (260:4cm) 
(0,0) -- (45:1cm) -- (90:6cm) (45:1cm) -- (0:6cm);
\node at (0,0) {};
\node[draw=none,fill=none,below right] at (0,0) {$x$};
\node at (45:1cm) {};
\node[draw=none,fill=none,above right] at (45:1cm) {$x_i$};
\path (a1) node  {};
\draw (a1) -- (50:6cm);
\path (a2) node {};
\draw (a2) -- (40:6cm);
\node[draw=none,fill=none] at (225:1.2cm) {$A_0^2$};
\node[draw=none,fill=none] at (225:3cm) {$A_0^4$};
\node[draw=none,fill=none] at (73:3.2cm) {$A_1^2$};
\node[draw=none,fill=none] at (17:3.2cm) {$A_2^2$};
  \end{scope}
  \begin{scope}[scale=.5,xshift=15cm]
\draw[dashed] (0,0) circle (1cm) circle (2cm) circle (3cm) circle (4cm) circle (5cm) circle (6cm);
\coordinate (a1) at (.535,1.92);
\coordinate (a2) at (1.92,.535);
\begin{scope}
  \clip (0,0) circle (4cm);
\clip (190:6cm) -- (0,0) -- (260:6cm);
\path[fill=gray!20] (190:6cm) -- (0,0) -- (260:6cm);
\end{scope}
\begin{scope}
  \clip (0,0) circle (2cm);
\clip (190:4cm) -- (0,0) -- (260:4cm);
\path[fill=gray!50] (190:4cm) -- (0,0) -- (260:4cm);
\end{scope}
\begin{scope}
  \clip (0,0) circle (4cm);
\clip (a1) -- (0,6) -- (50:6cm);
\path[fill=gray!50] (a1) -- (0,6) -- (50:6cm);
\end{scope}
\begin{scope}
  \clip (0,0) circle (4cm);
\clip (a2) -- (6,0) -- (40:6cm);
\path[fill=gray!50] (a2) -- (6,0) -- (40:6cm);
\end{scope}

\draw 
(0,0) -- (190:4cm) (0,0) -- (260:4cm) 
(0,0) -- (45:1cm) -- (90:6cm) (45:1cm) -- (0:6cm);
\node at (0,0) {};
\node[draw=none,fill=none,below right] at (0,0) {$y$};
\node at (45:1cm) {};
\node[draw=none,fill=none,above right] at (45:1cm) {$y_i$};
\path (a1) node  {};
\draw (a1) -- (50:6cm);
\path (a2) node {};
\draw (a2) -- (40:6cm);
\node[draw=none,fill=none] at (225:1.2cm) {$B_0^2$};
\node[draw=none,fill=none] at (225:3cm) {$B_0^4$};
\node[draw=none,fill=none] at (73:3.2cm) {$B_1^2$};
\node[draw=none,fill=none] at (17:3.2cm) {$B_2^2$};
  \end{scope}
\end{tikzpicture}  
  \caption{Proof of Lemma \protect\ref{lem:trees}:
The inductive step from $h=2$ to $h+2=4$ ($A_0^2\cong B_0^2\implies A_0^4\cong B_0^4$).
The dashed circles correspond to the metric in $T_x$ and $S_y$.
Note that $T_x^{4}(x_i)$ is rooted at $x_i$ and has branches $A_1^2$ and $A_2^2$.}
  \label{fig:ind-step}
\end{figure*}

The base case of $h=0$ is trivial. Note that it proves the lemma for $r=1$.
Let $r\ge2$ and $h\le r-2$. Assume that $A_0^{h}\cong B_0^{h}$ and deduce
from here that $A_0^{h+2}\cong B_0^{h+2}$
(for $h=r-2$ note that $A^r_0=A_0^{r-1}$ and $B^r_0=B_0^{r-1}$).
Via $\alpha_i$ we get $A_1^{h}\cong B_1^{h}$.
Since we also have $A_j^{h}\cong B_j^{h}$ for all $j\ge2$,
we can combine these isomorphisms 
and obtain an isomorphism
from $T_x^{h+2}(x_i)$ onto $S_y^{h+2}(y_i)$; see Fig.~\ref{fig:ind-step}.
The condition $T^{r-1}_x\cong S^{r-1}_y$ implies that $T_x^{h+2}\cong S_y^{h+2}$.
Since $T^{h+2}_x$ actually consists of $T_x^{h+2}(x_i)$ and $A_0^{h+2}$
and $S^{h+2}_y$ consists of $S_y^{h+2}(y_i)$ and $B_0^{h+2}$, we conclude
that $A_0^{h+2}$ and $B_0^{h+2}$ are also isomorphic.
\end{proof}

\subsection{A relationship between universal covers and color refinement }

A partition $\Pi$ of the vertex set of a graph $G$ is called \emph{equitable} if 
for any elements $X\subseteq V(G)$ and $Y\subseteq V(G)$ of $\Pi$
the following is true: Any two vertices
$u$ and $v$ in $X$ have the same number of neighbors in~$Y$.
A trivial example of an equitable partition is the partition of $V(G)$
into singletons. There is a unique equitable partition $\Pi^*$ that is
the \emph{coarsest} in the sense that any other equitable partition $\Pi$
is a subpartition of $\Pi^*$. 
This partition can be found by
the following \emph{color refinement} procedure.
Define a sequence of colorings $C^i$ of $G$ recursively.
The initial coloring $C^0$ is uniform. Then,
\begin{equation}
  \label{eq:Ci}
 C^{i+1}(u)=\left\langle C^i(u),\msetdef{C^i(a)}{a\in N(u)}\right\rangle, 
\end{equation}
where $\left\langle \ldots \right\rangle$ denotes an ordered pair and
$\left\{\!\!\left\{ \ldots \right\}\!\!\right\}$ denotes a multiset.
Regard $C^i$ as a partition of $V(G)$ (consisting of the monochromatic classes
of vertices) and note that $C^{i+1}$ is a refinement of $C^i$.
It follows that, eventually, $C^{s+1}=C^s$ for some $s$; hence, $C^{i}=C^s$ for any $i\ge s$.
Such a partition $C^s$ is called \emph{stable}. It is easy to see that $C^s$
is equitable, and an inductive argument shows that it is the coarsest; see~\cite[Lemma 1]{CardonC82}.
The minimum $s$ such that $C^s$ is stable will be denoted by~$\stabi G$.

We now show that the color $C^i(v)$ describes the isomorphism type
of the universal cover $U_v(G)$ truncated at depth~$i$.

\begin{lemma}\label{lem:UvsC}
  Let $U$ and $W$ be universal covers of graphs $G$ and $H$ respectively.
Furthermore, let $\alpha$ be a covering map from $U$ to $G$
and $\beta$ be a covering map from $W$ to $H$. Then
$U^i_x\cong W^i_y$ if and only if $C^i(\alpha(x))=C^i(\beta(y))$.
\end{lemma}

\begin{proof}
  We use induction on $i$. 
The base case of $i=0$ is trivial.

Assume that the claim is true for $i$. By Lemma \ref{lem:trees},
$U^{i+1}_x\cong W^{i+1}_y$ if and only if $U^i_x\cong W^i_y$
and $\msetdef{U^i_z}{z\in N(x)}=\msetdef{W^i_z}{z\in N(y)}$,
where the multisets consist of the isomorphism types of rooted trees.
By the induction assumption, the former condition is equivalent to
$C^i(\alpha(x))=C^i(\beta(y))$ and the latter condition is equivalent to
$\msetdef{C^i(\alpha(z))}{z\in N(x)}=\msetdef{C^i(\beta(z))}{z\in N(y)}$.
Since $\alpha$ and $\beta$ are bijective on $N(x)$ and $N(y)$,
the last equality can be rewritten as
$\msetdef{C^i(a)}{a\in N(\alpha(x))}=\msetdef{C^i(a)}{a\in N(\beta(y))}$.
By the definition of $C^{i+1}(u)$, we conclude that
$U^{i+1}_x\cong W^{i+1}_y$ if and only if $C^{i+1}(\alpha(x))=C^{i+1}(\beta(y))$.
\end{proof}

The next lemma shows an upper bound for the truncation depth
sufficient to detect isomorphism of universal covers.

\begin{lemma}\label{lem:twon}
Suppose that $G$ and $H$ are connected graphs with at most
$n$ vertices each. Let $U$ and $W$ be their universal covers. 
Then $U_x\cong W_y$ if and only if $U_x^{2n-1}\cong W_y^{2n-1}$.  
\end{lemma}

\begin{proof}
  One direction is trivial. Suppose that $U_x^{2n-1}\cong W_y^{2n-1}$.
Fix covering maps $\alpha$ of $U$ onto $G$ and $\beta$ of $W$ onto $H$.
By Lemma \ref{lem:UvsC}, $C^{2n-1}(\alpha(x))=C^{2n-1}(\beta(y))$.
Note that $C^{2n-1}$ is a stable partition of the disjoint union of $G$ and $H$.
It follows that $C^i(\alpha(x))=C^i(\beta(y))$ for all $i$.
Using Lemma \ref{lem:UvsC} again, we conclude that $U_x^i\cong W_y^i$ for all $i$.
Consider the set of all isomorphisms from $U_x^i$ to $W_y^i$, $i\ge0$.
Define an auxiliary graph on this set joining an isomorphism $\phi$ from 
$U_x^i$ to $W_y^i$ and an isomorphism $\psi$ from $U_x^{i+1}$ to $W_y^{i+1}$
by an edge if $\psi$ extends $\phi$. By König's lemma, this graph
contains an infinite path $\phi_0\subset\phi_1\subset\phi_2\subset\ldots$.
The union of these isomorphisms gives us an isomorphism from $U_x$ to~$W_y$.
\end{proof}

\subsection{Existence of a common cover: Angluin's algorithm}\label{ss:angluin}

\begin{lemma}\label{lem:common}
  Suppose that connected graphs $G$ and $H$ have at most $n$ vertices each.
Then the following conditions are equivalent:
\begin{enumerate}
\item[(1)]
$G$ and $H$ have a common covering graph.
\item[(2)]
\mbox{$\setdef{C^{2n-1}(u)}{u\in V(G)}\cap\setdef{C^{2n-1}(v)}{v\in V(H)}\ne\emptyset$.}
\item[(3)]
$\setdef{C^{2n-1}(u)}{u\in V(G)}=\setdef{C^{2n-1}(v)}{v\in V(H)}$.
\item[(4)]
$\setdef{C^{s+1}(u)}{u\in V(G)}=\setdef{C^{s+1}(v)}{v\in V(H)}$, where $s=\stabi G$.
\end{enumerate}
\end{lemma}

Note that Condition 4 improves Condition 3 because $s+1\le n$.
While the equivalence to Conditions 2 and 3 means that
Condition 1 can be detected after stabilization of the coloring $C^i$
on the disjoint union of $G$ and $H$, the equivalence to Condition 4
means that this can actually be done as soon as 
$C^i$ stabilizes at least on \emph{one} of the graphs $G$ and~$H$.
Another consequence of Lemma \ref{lem:common} is this:
If $G$ and $H$ do not share a covering graph, then the color sets
$\setdef{C^{i}(u)}{u\in V(G)}$ and $\setdef{C^{i}(v)}{v\in V(H)}$
must be unequal starting from $i=n$ and disjoint starting from $i=2n-1$.

\begin{proof}
\textbf{$\mathbf{(1)\iff(2)}$.}
By Lemma \ref{lem:UGUH}, $G$ and $H$ have a common cover
if and only if their universal covers $U$ and $W$ are isomorphic.
Note that $U\cong W$ iff $U_x\cong W_y$ for some $x$ and $y$.
By Lemmas \ref{lem:twon} and \ref{lem:UvsC}, 
$$
U_x\cong W_y\iff U_x^{2n-1}\cong W_y^{2n-1}\iff
C^{2n-1}(\alpha(x))=C^{2n-1}(\beta(y)),
$$
for covering maps $\alpha$ from $U$ to $G$ and $\beta$ from $W$ to $H$.
It follows that $G$ and $H$ have a common cover
iff $C^{2n-1}(u)=C^{2n-1}(v)$ for some $u\in V(G)$ and $v\in V(H)$, which
is exactly Condition 2 in the lemma.

\smallskip

\textbf{$\mathbf{(2)\iff(3)}$.}
Condition 3 implies Condition 2 and follows by Lemma \ref{lem:UvsC} 
from the isomorphism $U\cong W$, which is implied by Condition~2.

\smallskip

\textbf{$\mathbf{(3)\iff(4)}$.}
Consider $C^{s}$, which is stable on $G$.
If the sets
$
\setdef{C^{s}(u)}{u\in V(G)}$ and
$\setdef{C^{s}(v)}{v\in V(H)}
$
are not equal,
then Conditions 3 and 4 are both false because unequal vertex colors cannot become equal later.
Suppose that these sets are equal
and compare the partitions $C^{s}$ and $C^{s+1}$ of $V(G)\cup V(H)$.
If $C^s\ne C^{s+1}$, this means that 
the number of colors in the $C^{s+1}$-coloring of $H$ is strictly larger
than in the $C^{s+1}$-coloring of $G$.
Therefore, Condition 4 is false, and Condition 3 is false too. 
If $C^s=C^{s+1}$, this means that the partition stabilizes on $V(G)\cup V(H)$.
In this case, vertices equally colored with respect to $C^{s}$
remain equally colored with respect to any $C^i$, which implies
that both Conditions 3 and 4 are true.
\end{proof}

Lemma \ref{lem:common} justifies Angluin's algorithm \cite{Angluin80} 
for deciding in polynomial time whether or not two given graphs
have a common covering graph.
Any of Conditions 2, 3, or 4 can be checked efficiently.
Note that we cannot compute the colors $C^{2n-1}(u)$ literally as they defined by Equation \refeq{Ci}
because the length of $C^i(u)$ grows exponentially with $i$.
We can overcome this complication by renaming the colors after each refinement step
(doing so, we will never need more than $2n$ color names).

\section{Norris's problem}\label{s:Norris}

In order to characterize expressibility of two-variable logic with counting quantifiers, 
Immerman and Lander \cite{ImmermanL90}
introduced the \emph{2-pebble counting game} on graphs $G$ and $H$.
We need a restricted version of this game 
that was defined and used for diverse purposes in \cite{GraedelO99,AtseriasM13}.
We call it the \emph{counting bisimulation game}.
The game is played by two players, Spoiler and Duplicator,
to whom we will refer as he and she respectively.
Each player has a pair of distinct pebbles $p$ and $q$, that
are put on vertices of $G$ and $H$ during the game.
The two copies of the same pebble are always in different graphs.
In each \emph{round} of the game the copies of one of the pebbles $p$ and $q$
change their positions while the copies of the other pebble do not move.
A round where the pebble $p$ moves is played as follows.
First, Spoiler removes the copies of $p$ from the board (if they are already put there).
Suppose that the copies of $q$ are located on vertices $u\in V(G)$ and $v\in V(H)$ at this point
of time.
Then Spoiler specifies a set of vertices $A$ in one of the graphs such that
either $A\subseteq N(u)$ or $A\subseteq N(v)$.
Duplicator has to respond with a set of vertices $B$
in the other graph such that $B\subseteq N(v)$ or $B\subseteq N(u)$ respectively.
Duplicator must keep the condition $|B|=|A|$ true; otherwise she loses the game in this round.
Finally, Spoiler puts $p$ on a vertex $b\in B$.
In response Duplicator has to place the other copy
of this pebble on a vertex $a\in A$.
In the next round the roles of $p$ or $q$ are interchanged.
We write $\game^r(G,u,H,v)$ to denote the $r$-round counting bisimulation game
on $G$ and $H$ starting from the position where 
the vertices $u\in V(G)$ and $v\in V(H)$ are occupied by one of the pebbles, say,
by the copies of~$q$. Duplicator wins this game if she does not lose any of the $r$ rounds.

\begin{lemma}\label{lem:game-color}
  Duplicator has a winning strategy in $\game^i(G,u,H,v)$
if and only if $C^i(u)=C^i(v)$.
\end{lemma}

\begin{proof}
We use induction on $i$. For $i=0$ the claim is trivially true.
Assume that it is true for some $i\ge0$.

Suppose that
 $C^{i+1}(u)=C^{i+1}(v)$.
Then Duplicator wins $\game^{i+1}(G,u,H,v)$ as follows.
Without loss of generality, suppose that in the first round Spoiler specifies a set $A\subseteq N(u)$.
It follows from 
$C^{i+1}(u)=C^{i+1}(v)$
that 
\begin{equation}
  \label{eq:mseteq}
  \msetdef{C^i(a)}{a\in N(u)}=\msetdef{C^i(b)}{b\in N(v)}.
\end{equation}
This means that Duplicator can ``mirror'' the set $A$ with a set $B\subseteq N(v)$ such that
\refeq{mseteq} stays true if $N(u)$ is restricted to $A$ and
$N(v)$ is restricted to $B$. This allows Duplicator to ensure pebbling
vertices $a\in V(G)$ and $b\in V(H)$ such that $C^i(a)=C^i(b)$
and wins the remaining $i$ rounds by the induction assumption.

Suppose now that
  $C^{i+1}(u)\ne C^{i+1}(v)$.
Then Spoiler wins $\game^{i+1}(G,u,H,v)$ as follows.
The inequality of colors $C^{i+1}(u)$ and $C^{i+1}(v)$ means that either
$
C^{i}(u)\ne C^{i}(v)
$
or
$
\msetdef{C^i(a)}{a\in N(u)}\ne\msetdef{C^i(b)}{b\in N(v)}.
$
In the former case Spoiler wins in $i$ moves by the induction assumption.
In the latter case, there is a $C^i$-color $c$ such that for one of the two
vertices, say, for $u$, the set of neighbors $a$ with $C^i(a)=c$ is
strictly larger than the analogous set for $v$. Specifying the larger set
as $A$, Spoiler ensures pebbling $a$ and $b$ such that $C^i(a)\ne C^i(b)$
and wins in the next $i$ moves by the induction assumption.
\end{proof}

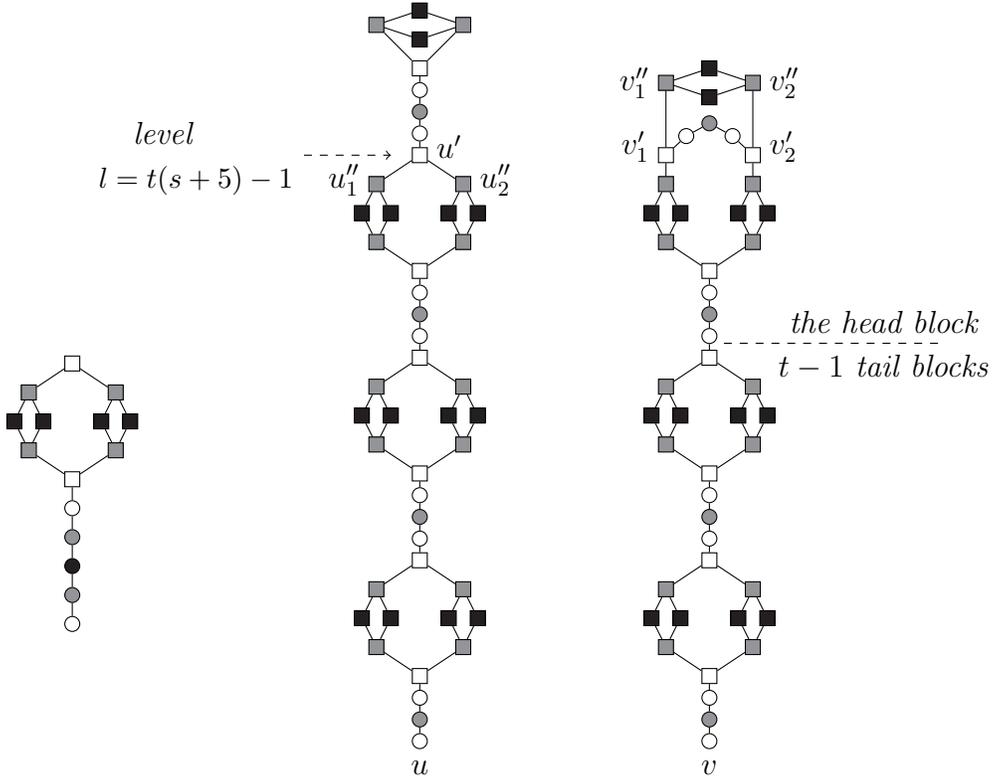
\begin{figure*}
  \centering
\begin{tikzpicture}[scale=.77]
\newcommand{\diamoond}[3]{
\path (#1,#2) node[bvertex] (b#3) {}
 (#1-.5,#2+1) node[cvertex] (c#3) {} edge (b#3)
 (#1+.5,#2+1) node[cvertex] (cc#3) {} edge (b#3) 
    (#1,#2+2) node[bvertex] (bb#3) {} edge (c#3) edge (cc#3);
}
\newcommand{\bloock}[3]{
\diamoond{#1-1.5}{#2+3}{1}
\diamoond{#1+1.5}{#2+3}{2}
\path (#1,#2-.25) node[vertex] (v#3) {}
    (#1,#2+.5) node[vvertex] (vnew) {} edge (v#3)
    (#1,#2+1.25) node[vertex] (vv) {} edge (vnew)
    (#1,#2+2) node[avertex] (a) {} edge (b1) edge (b2) edge (vv)
    (#1,#2+6) node[avertex] (a#3) {} edge (bb1) edge (bb2);
}
\begin{scope}[scale=.5,xshift=-70mm,yshift=38mm]
\path (0,0) node[vertex] (v1) {}
      (0,1) node[vvertex] (vv1) {} edge (v1)
      (0,2) node[vvvertex] (vvv) {} edge (vv1)
      (0,3) node[vvertex] (vv2) {} edge (vvv)
      (0,4) node[vertex] (v2) {} edge (vv2)
      (0,5) node[avertex] (a1) {} edge (v2)
   (-1.5,6) node[bvertex] (b11) {} edge (a1)
    (1.5,6) node[bvertex] (b12) {} edge (a1)
     (-2,7) node[cvertex] (c1) {} edge (b11)
     (-1,7) node[cvertex] (c2) {} edge (b11) 
      (1,7) node[cvertex] (c3) {} edge (b12)
      (2,7) node[cvertex] (c4) {} edge (b12) 
   (-1.5,8) node[bvertex] (b21) {} edge (c1) edge (c2)
    (1.5,8) node[bvertex] (b22) {} edge (c3) edge (c4)
      (0,9) node[avertex] (a2) {} edge (b21) edge (b22);
\end{scope}

\begin{scope}[scale=.5,xshift=50mm]
\bloock001
\bloock072
\bloock0{14}3
\path (0,20.75) node[vertex] (v4) {}
      (0,21.5) node[vvertex] (vnew45) {} edge (v4)
      (0,22.25) node[vertex] (v5) {} edge (vnew45)
      (0,23) node[avertex] (a4) {} edge (v5)
   (-1.5,24.5) node[bvertex] (b4) {} edge (a4)
    (1.5,24.5) node[bvertex] (bb4) {} edge (a4)
    (0,24) node[cvertex] (c4) {} edge (b4) edge (bb4)
    (0,25) node[cvertex] (c5) {} edge (b4) edge (bb4);
\draw (v2) -- (a1) (v3) -- (a2) (v4) -- (a3);
\node[draw=none,fill=none,below] at (0,-.55) {$u$};
\node[draw=none,fill=none,right] at (.2,20.3) {$u'$};
\node[draw=none,fill=none,left] at (-1.7,19.15) {$u''_1$};
\node[draw=none,fill=none,right] at (1.7,19.15) {$u''_2$};
\draw[->,dashed] (-4,20) -- (-1,20);
\node[draw=none,fill=none,above left] at (-7.5,20) {\it level};
\node[draw=none,fill=none,below left] at (-4,20) {\small\it $l=t(s+5)-1$};
\end{scope}

\begin{scope}[scale=.5,xshift=150mm]
\bloock001
\bloock072
\diamoond{-1.5}{17}x
\diamoond{1.5}{17}y
\path (0,13.75) node[vertex] (v3) {} edge (a2) 
      (0,14.5) node[vvertex] (vnew3) {} edge (v3)
      (0,15.25) node[vertex] (vv3) {} edge (vnew3)
      (0,16) node[avertex] (aaa) {} edge (vv3) edge (bx) edge (by)
   (-1.5,20) node[avertex] (aaa1) {} edge (bbx)
    (1.5,20) node[avertex] (aaa2) {} edge (bby)
    (-.8,20.66) node[vertex] (vvv1) {} edge (aaa1)
     (0,21.1) node[vvertex] (vvvnew) {} edge (vvv1)
     (.8,20.66) node[vertex] (vvv2) {} edge (aaa2) edge (vvvnew)
    (-1.5,22.5) node[bvertex] (b4) {} edge (aaa1)
     (1.5,22.5) node[bvertex] (bb4) {} edge (aaa2)
     (0,22) node[cvertex] (c4) {} edge (b4) edge (bb4)
     (0,23) node[cvertex] (c5) {} edge (b4) edge (bb4);
\draw (v2) -- (a1);
\draw[dashed] (.5,13.5) -- (8,13.5);
\node[draw=none,fill=none,above] at (6,13.5) {\it the head block};
\node[draw=none,fill=none,below] at (6,13.5) {\it $t-1$ tail blocks};
\node[draw=none,fill=none,below] at (0,-.55) {$v$};
\node[draw=none,fill=none,left] at (-1.7,20.3) {$v'_1$};
\node[draw=none,fill=none,right] at (1.7,20.3) {$v'_2$};
\node[draw=none,fill=none,left] at (-1.7,22.5) {$v''_1$};
\node[draw=none,fill=none,right] at (1.7,22.5) {$v''_2$};
\end{scope}
\end{tikzpicture}  
  \caption{The tail block $B_s$ for $s=5$ and the graphs $G_{s,t}$ and $H_{s,t}$ for $s=3$, $t=3$.}
  \label{fig:GstHst}
\end{figure*}

\begin{theorem}\label{thm:nancy}
  For each $n$, there are $n$-vertex graphs $G$ and $H$
whose universal covers $U$ and $W$ contain
vertices $x$ and $y$ such that
$U^i_x\cong W^i_y$ for all $i\le2n-16\sqrt n$
while $U_x\not\cong W_y$.
\end{theorem}

To prove this result, consider graphs $G_{s,t}$ and $H_{s,t}$
as shown in Fig.~\ref{fig:GstHst}.
Each of the graphs is a chain of $t$ blocks, one head block
and $t-1$ tail blocks. All tail blocks in both $G_{s,t}$ and $H_{s,t}$
are copies of the same graph $B_s$ with $s+10$ vertices
(where $s$ is the number of vertices of degree at most 2
spanning a path in $B_s$). 
The head blocks of $G_{s,t}$ and $H_{s,t}$ are different. 
Note that the head block of $G_{s,t}$ also contains a copy of $B_s$.
Both head blocks have $2s+15$ vertices. Thus, both $G_{s,t}$ and $H_{s,t}$
have 
$$
n=(t+1)(s+10)-5
$$ 
vertices.
Both graphs have a single vertex of degree 1; we 
denote these vertices by $u$ and $v$ respectively.

The graphs $G_{s,t}$ and $H_{s,t}$ are uncolored. However, we distinguish
$\lceil s/2\rceil+3$ types of vertices in them, that are presented in Fig.~\ref{fig:GstHst}
by auxiliary colors and shapes.

We will use the following properties of $G_{s,t}$ and $H_{s,t}$.
\begin{description}
\item[(A)] 
The partition of $V(G_{s,t})\cup V(H_{s,t})$ by types is almost equitable:
With the exception of $u$ and $v$, any two vertices of the same type have
the same number of neighbors of each type. For example, if $s=3$, all
possible neighborhoods are these:
$$
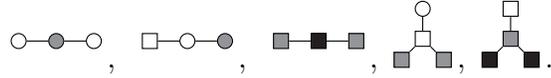
\raisebox{2.5mm}{
\begin{tikzpicture}[scale=.5]
\path (0,0) node[vertex] (a) {}
      (1,0) node[vvertex] (b)  {} edge (a)
      (2,0) node[vertex] (c)  {} edge (b);
\end{tikzpicture}}\,,\
\raisebox{2.5mm}{
\begin{tikzpicture}[scale=.5]
\path (0,0) node[avertex] (a) {}
      (1,0) node[vertex] (b)  {} edge (a)
      (2,0) node[vvertex] (c)  {} edge (b);
\end{tikzpicture}}\,,\
\raisebox{2.5mm}{
\begin{tikzpicture}[scale=.5]
\path (0,0) node[bvertex] (a) {}
      (1,0) node[cvertex] (b)  {} edge (a)
      (2,0) node[bvertex] (c)  {} edge (b);
\end{tikzpicture}}\,,\
\begin{tikzpicture}[scale=.4]
\path (-.7,-.7) node[bvertex] (a) {}
      (0,0) node[avertex] (b)  {} edge (a)
      (.7,-.7) node[bvertex] (c)  {} edge (b)
      (0,1) node[vertex] (d)  {} edge (b);
\end{tikzpicture}\,,\
\begin{tikzpicture}[scale=.4]
\path (-.7,-.7) node[cvertex] (a) {}
      (0,0) node[bvertex] (b)  {} edge (a)
      (.7,-.7) node[cvertex] (c)  {} edge (b)
      (0,1) node[avertex] (d)  {} edge (b);
\end{tikzpicture}\,.
$$
\item[(B)] 
Given a vertex $z$ in $G_{s,t}$ or $H_{s,t}$, we define its \emph{level} $\ell(z)$
as the distance from $z$ to $u$ or $v$ respectively.
Up to the level 
\begin{equation}
  \label{eq:ldef}
l=t(s+5)-1,  
\end{equation}
the following holds true: All vertices at the same level have the same type.
\end{description}

\begin{lemma}\label{lem:GstHst}
Consider the counting bisimulation $\game^r(G_{s,t},u,H_{s,t},v)$.
\begin{enumerate}
\item[\bf 1.] 
Spoiler has a winning strategy for $r=2l+1$;
\item[\bf 2.] 
Duplicator has a winning strategy for any $r\le2l$, where $l$ is defined by~\refeq{ldef}.
\end{enumerate}
\end{lemma}

\begin{proof}
 {\bf 1.} 
Let $a_0=u$ and $b_0=v$. As long as $i\le l$, in the $i$-th round 
Spoiler specifies $A$ consisting of a single vertex $a_i$ such that
$a_i$ and $a_{i-1}$ are adjacent and $\ell(a_i)=\ell(a_{i-1})+1$.
Let $B=\{b_i\}$ be Duplicator's response. 
As long as $\ell(b_i)=\ell(b_{i-1})+1$, we have 
\begin{equation}
  \label{eq:ellabi}
\ell(a_i)=\ell(b_i)=i.  
\end{equation}
This equality can be broken only if 
$\ell(b_i)=\ell(b_{i-1})-1$. 
In this case $\ell(b_i)=i-2$ while $\ell(a_i)=i$,
and Spoiler wins in the other $i-1$ moves by pebbling vertices in $H_{s,t}$ along 
a shortest path from $b_i$ to $v$. He needs $i-2$ moves to reach $v$ and wins in 
one extra move because the corresponding vertex in $G_{s,t}$ has degree larger than~1.

We, therefore, assume that the condition \refeq{ellabi} holds true for all $i\le l$
In the notation of Fig.~\ref{fig:GstHst}, this means that $a_l=u'$ and $b_l=v'_j$ for $j=1$ or $j=2$.
In the $(l+1)$-th round Spoiler specifies the set $A=\{u''_1,u''_2\}$
consisting of the two \vtype{bvertex}-neighbors of $u'$ (both at the level $l-1$). If Duplicator's response $B$
contains the \vtype{vertex}-neighbor of $v'_j$, Spoiler pebbles it and wins in the next round
because the vertices pebbled in $G_{s,t}$ and $H_{s,t}$ have different degrees.
Otherwise, $B$ consists of the two \vtype{bvertex}-neighbors of $v'_j$. One of them,
$v''_j$, is at the level $l+1$. Spoiler pebbles it, while Duplicator has to pebble
$u''_k$ for $k=1$ or $k=2$. Using the fact that $\ell(u''_k)=l-1$ and $\ell(v''_j)=l+1$,
Spoiler wins in the next $l$ rounds by pebbling vertices along a shortest path from $u''_k$ to $u$.
At total, he makes $2l+1$ moves to win.

\smallskip

{\bf 2.} 
As before, let $a_i$ and $b_i$ denote the vertices of $G_{s,t}$ and $H_{s,t}$ pebbled in
the $i$-th round. Note that Spoiler can win in the next $(i+1)$-th round only if
$a_i$ and $b_i$ have different degrees. A sufficient condition for $a_i$ and $b_i$ having
equal degrees is that these vertices are of the same type
(excepting the case that one of them is $u$ or $v$). Property A of $G_{s,t}$ and $H_{s,t}$
implies that, if $a_{i-1}$ and $b_{i-1}$ have the same type,
Duplicator can ensure the same for $a_i$ and $b_i$ unless $a_{i-1}=u$ or $b_{i-1}=v$
(the case that both $a_{i-1}=u$ and $b_{i-1}=v$ is also favorable for Duplicator).
This observation is the basis of Duplicator's strategy. In the first phase of the game,
Duplicator keeps the levels of $a_i$ and $b_i$ equal. Up to level $l$,
by property B this implies also the equality of types.

More precisely, we fix a strategy for Duplicator such that
\begin{enumerate}
\item 
$\ell(a_i)=\ell(b_i)$ as long as 
\begin{equation}
  \label{eq:<l}
\ell(a_{i-1})=\ell(b_{i-1})<l;
\end{equation}
\item 
$a_i$ and $b_i$ have the same type as long as 
\begin{equation}
  \label{eq:>0}
\min\{\ell(a_{i-1}),\ell(b_{i-1})\}>0\text{ or }\ell(a_{i-1})=\ell(b_{i-1})=0.
\end{equation}
\end{enumerate}
Duplicator can keep Conditions 1 and 2 true due to Properties A and B of $G_{s,t}$ and $H_{s,t}$.
Note that \refeq{<l} implies \refeq{>0} and that fulfilling Condition 1 implies also
fulfilling Condition 2. Thus, Duplicator is alive as long as \refeq{>0} holds true.

Note now that \refeq{<l} can be broken for $i-1=l$ at the earliest (i.e., not earlier than in the $l$-th round) 
and, when this happens, we will have $\ell(a_{i-1})=\ell(b_{i-1})=l$.
Starting from this point, Spoiler needs no less than $l$ moves to break \refeq{>0}.
Therefore, Duplicator survives at least during the first $2l$ rounds
irrespectively of Spoiler's strategy.
\end{proof}

Turning back to the proof of Theorem \ref{thm:nancy}, 
consider $G=G_{s,t}$ and $H=H_{s,t}$, where we set $s=2t+1$.
Thus, both $G$ and $H$ have $n=(t+1)(s+10)-5=2t^2+13t+6$ vertices.
By Lemma \ref{lem:GstHst}.2, Duplicator has a winning strategy in
the counting bisimulation $\game^r(G,u,H,v)$ for 
$$
r=2t(s+5)-2=4t^2+12t-2=2n-14t-14>2n-16\sqrt n.
$$
By Lemma \ref{lem:game-color}, $C^r(u)=C^r(v)$. Consider covering maps
$\alpha$ from $U$ to $G$ and $\beta$ from $W$ to $H$.
Take $x$ and $y$ such that $\alpha(x)=u$ and $\beta(y)=v$.
By Lemma \ref{lem:UvsC}, $U^r_x\cong W^r_y$.
On the other hand, combining Lemmas \ref{lem:UvsC} and \ref{lem:game-color}
with Lemma \ref{lem:GstHst}.1, we conclude that $U^{r+1}_x\not\cong W^{r+1}_y$,
which implies that $U_x\not\cong W_y$. This proves the theorem in the case
that $n=2t^2+13t+6$ for some~$t$. 

For any other $n$, fix $t$ such that $2t^2+13t+6<n<2(t+1)^2+13(t+1)+6$.
Now we consider $G$ obtained from $G_{2t+1,t}$ by adding new $k=n-(2t^2+13t+6)$
vertices and connecting each of them to $u$ by an edge. The graph $H$
is obtained from $H_{2t+1,t}$ in the same way. 
The new vertices do not affect the outcome of $\game^r(G,u,H,v)$
if $k>3$. If $k\le3$, Duplicator can resist at most one round longer.
Since $r+1=4t^2+12t-1$ is still larger than $2n-16\sqrt n$,
the proof is complete.

\begin{remark}\label{rem:orgraphs}\rm
 Theorem \ref{thm:nancy} is true also for oriented graphs.
To see this, let us orient the graphs $G$ and $H$ constructed in the proof.
Note that adjacent vertices in these graphs have different types.
Fix arbitrarily an order on the set of all vertex types.
For each pair of adjacent vertices $a$ and $b$, we draw an arrow from $a$
to $b$ if the type of $b$ has higher position with respect to this order
than the type of $a$, for example,
\begin{equation}
  \label{eq:orient}
  \begin{tikzpicture}[scale=.5]
\path (0,0) node[vvertex] (a) {}
      (1,0) node[vertex] (b)  {} edge[<-] (a);
\end{tikzpicture}\,,\quad
\begin{tikzpicture}[scale=.5]
\path (0,0) node[vertex] (a) {}
      (1,0) node[avertex] (b)  {} edge[<-] (a);
\end{tikzpicture}\,,\quad
\begin{tikzpicture}[scale=.5]
\path (0,0) node[avertex] (a) {}
      (1,0) node[bvertex] (b)  {} edge[<-] (a);
\end{tikzpicture}\,,\quad
\begin{tikzpicture}[scale=.5]
\path (0,0) node[bvertex] (a) {}
      (1,0) node[cvertex] (b)  {} edge[<-] (a);
\end{tikzpicture}\,.
\end{equation}
(assuming $s=3$ as in Fig.~\ref{fig:GstHst}).
Denote the corresponding oriented graphs by $G'$ and $H'$.
Let $U=U_u(G)$ and $W=U_v(H)$. The vertices of $U$ and $W$
inherit the types in a natural way. Orient $U$ and $W$
accordingly to \refeq{orient} and denote the resulting 
oriented graphs by $U'$ and $W'$. Note that $U'$ and $W'$
are universal covers of $G'$ and $H'$ respectively.
Furthermore, isomorphism of $U$ and $W$ truncated at depth $i$
implies isomorphism of $U'$ and $W'$ truncated at the same depth.
\end{remark}

\begin{remark}\label{rem:port-num}\rm
With an additional effort, an analog of Theorem \ref{thm:nancy}
can be obtained also for port-numbered graphs, which are
a popular model of distributed networks \cite{Angluin80,DereniowskiKP14,Hendrickx13,YamashitaK88}.
Formally, a port-numbered graph can be defined as a relational structure that is
an undirected graph $G$ of maximum degree $D$ where, for each edge $\{x,y\}$, 
each of the ordered pairs $(x,y)$
and $(y,x)$ satisfies exactly one of the binary relations $B_1,\ldots,B_D$.
For every vertex $x$ of degree $d$, each of the relations $B_1,\ldots,B_d$
must be satisfied on some $(x,y)$. Thus, all arrows $(x,y)$
emanating from $x$ satisfy pairwise different relations $B_i$ for $i\le d$.
The meaning of $B_i(x,y)$ is that the edge $\{x,y\}$ receives the number $i$
among the edges incident to $x$. Since the techniques used in the proof of Theorem \ref{thm:nancy}
easily generalize to arbitrary binary structures, we can proceed basically in the same way.

We first construct undirected graphs and then
endow them with a suitable port-numbering. Specifically,
we construct graphs $G_{s,t}$ and $H_{s,t}$ similarly
to the proof of Theorem \ref{thm:nancy} from the tail and the
head blocks depicted in Fig.~\ref{fig:port}.
The types of vertices shown in the picture are now
considered to be colors of vertices. Thus,
$G_{s,t}$ and $H_{s,t}$ are vertex-colored graphs.
We define $G=G_{2s,t}$ and $H=H_{2s,t}$; note that the first
parameter is now even.

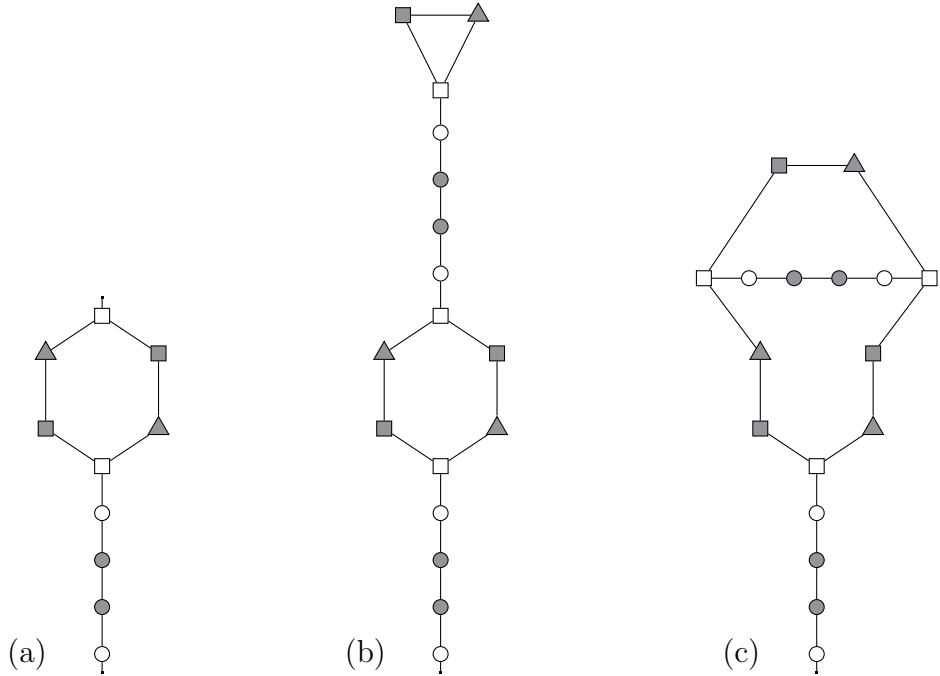
\begin{figure*}
  \centering
\begin{tikzpicture}
\newcommand{\diamoond}[3]{
\path (#1,#2) node[bvertex] (bb#3) {}
    (#1,#2+2) node[bbvertex] (bt#3) {} edge (bb#3);
}
\newcommand{\diamoondrev}[3]{
\path (#1,#2) node[bbvertex] (bb#3) {}
    (#1,#2+2) node[bvertex] (bt#3) {} edge (bb#3);
}
\newcommand{\tail}[3]{
\path (#1,#2) node[vertex] (vb#3) {}
      (#1,#2+1.25) node[vvertex] (vvb#3) {} edge (vb#3)
      (#1,#2+2.5) node[vvertex] (vvt#3) {} edge (vvb#3)
      (#1,#2+3.75) node[vertex] (vt#3) {} edge (vvt#3);
}
\newcommand{\bloock}{
\tail001
\diamoond{-1.5}61
\diamoondrev{1.5}62
\path (0,5) node[avertex] (ab) {} edge (vt1) edge (bb1) edge (bb2);
}

\begin{scope}[scale=.5]
\bloock
\path (0,9) node[avertex] (at) {} edge (bt1) edge (bt2)
      (0,9.5) node[inner sep=0pt,fill=black] () {} edge (at)
      (0,-.5) node[inner sep=0pt,fill=black] () {} edge (vb1);
\node[draw=none,fill=none] at (-2,0) {(a)};
\end{scope}

\begin{scope}[scale=.5,xshift=90mm]
\bloock
\tail0{10.125}2
\path (0,9) node[avertex] (at) {} edge (bt1) edge (bt2) edge (vb2)
      (0,15) node[avertex] (atop) {} edge (vt2) 
      (-1,17) node[bvertex] (bbl) {} edge (atop)
      (1,17) node[bbvertex] (bbr) {} edge (atop) edge (bbl)
(0,-.5) node[inner sep=0pt,fill=black] () {} edge (vb1);
\node[draw=none,fill=none] at (-2,0) {(b)};
\end{scope}

\begin{scope}[scale=.5,xshift=190mm]
\bloock
\path (-3,10) node[avertex] (al) {} edge (bt1)
      (-1.8,10) node[vertex] (vl) {} edge (al)
      (-.6,10) node[vvertex] (vvl) {} edge (vl)
      (.6,10) node[vvertex] (vvr) {} edge (vvl)
      (1.8,10) node[vertex] (vr) {} edge (vvr)
      (3,10) node[avertex] (ar) {} edge (bt2)  edge (vr)
     (-1,13) node[bvertex] (bbl) {} edge (al)
      (1,13) node[bbvertex] (bbr) {} edge (ar) edge (bbl)
(0,-.5) node[inner sep=0pt,fill=black] () {} edge (vb1);
\node[draw=none,fill=none] at (-2,0) {(c)};
\end{scope}
\end{tikzpicture}  
  \caption{(a) The tail block $B_4$.
(b) The head block of $G_{4,t}$.
(c) The head block of~$H_{4,t}$.}
  \label{fig:port}
\end{figure*} 

The modified construction ensures
some special properties of the color partition that
will be used below. The first useful property is the same as
in the proof of Theorem \ref{thm:nancy}: The partition is almost
equitable (except for the bottom vertices $u$ and $v$).
This allows us to prove that 
\begin{equation}\label{eq:iso-r}
U_u^r(G)\cong U_v^r(H)\text{ for }r=2n-O(\sqrt n).  
\end{equation} 

On the other hand, $U_u(G)\not\cong U_v(H)$. In fact,
even a stronger statement is true. Let $G_0$ and $H_0$ denote
the uncolored versions of $G$ and $H$ respectively.
Like in the proof of Theorem \ref{thm:nancy}, we can prove that
\begin{equation}
  \label{eq:non-iso}
U_u(G_0)\not\cong U_v(H_0).  
\end{equation}
This actually follows from the
fact that the eccentricity of the vertex $u$ in $G_0$ is greater
than the eccentricity of $v$ in~$H_0$. 

Now, we convert $G$ and $H$ into port-numbered graphs as follows.
Fix an arbitrary order on the vertex colors, for example,
\begin{equation}
  \label{eq:color-order}
\begin{tikzpicture}
\node[vertex] at (0,0) {};
\end{tikzpicture}\,,\quad
\begin{tikzpicture}
\node[vvertex] at (0,0) {};
\end{tikzpicture}\,,\quad
\begin{tikzpicture}
\node[avertex] at (0,0) {};
\end{tikzpicture}\,,\quad
\begin{tikzpicture}
\node[bvertex] at (0,0) {};
\end{tikzpicture}\,,\quad
\begin{tikzpicture}
\node[bbvertex] at (0,0) {};
\end{tikzpicture}  
\end{equation}
(we assume $s=4$ as in Fig.~\ref{fig:port}).
Notice another useful property of the vertex coloring:
The neighbors of each vertex $x$ have pairwise different colors.
This allows us to unambiguously enumerate the edges incident to $x$
consistently with the order \refeq{color-order}. In our example, this results in
the following labeling:
$$
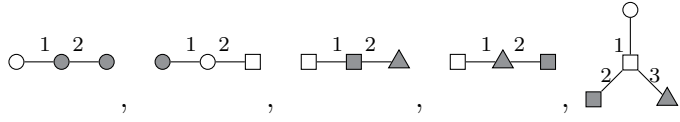
\raisebox{5mm}{
\begin{tikzpicture}[scale=.6]
\path (0,0) node[vertex] (a) {}
      (1,0) node[vvertex] (b)  {} edge (a)
      (2,0) node[vvertex] (c)  {} edge (b)
       (b) node[draw=none,fill=none,above left] {\scriptsize 1}
       (b) node[draw=none,fill=none,above right] {\scriptsize 2};
\end{tikzpicture}}\,,\
\raisebox{5mm}{
\begin{tikzpicture}[scale=.6]
\path (0,0) node[vvertex] (a) {}
      (1,0) node[vertex] (b)  {} edge (a)
      (2,0) node[avertex] (c)  {} edge (b)
       (b) node[draw=none,fill=none,above left] {\scriptsize 1}
       (b) node[draw=none,fill=none,above right] {\scriptsize 2};
\end{tikzpicture}}\,,\
\raisebox{5mm}{
\begin{tikzpicture}[scale=.6]
\path (0,0) node[avertex] (a) {}
      (1,0) node[bvertex] (b)  {} edge (a)
      (2,0) node[bbvertex] (c)  {} edge (b)
       (b) node[draw=none,fill=none,above left] {\scriptsize 1}
       (b) node[draw=none,fill=none,above right] {\scriptsize 2};
\end{tikzpicture}}\,,\
\raisebox{5mm}{
\begin{tikzpicture}[scale=.6]
\path (0,0) node[avertex] (a) {}
      (1,0) node[bbvertex] (b)  {} edge (a)
      (2,0) node[bvertex] (c)  {} edge (b)
       (b) node[draw=none,fill=none,above left] {\scriptsize 1}
       (b) node[draw=none,fill=none,above right] {\scriptsize 2};
\end{tikzpicture}}\,,\
\begin{tikzpicture}[scale=.7]
\path (-.7,-.7) node[bvertex] (a) {}
      (0,0) node[avertex] (b)  {} edge (a)
      (.7,-.7) node[bbvertex] (c)  {} edge (b)
      (0,1) node[vertex] (d)  {} edge (b)
       (.1,0) node[draw=none,fill=none,above left] {\scriptsize 1}
       (-.15,.1) node[draw=none,fill=none,below left] {\scriptsize 2}
       (.15,.1) node[draw=none,fill=none,below right] {\scriptsize 3};
\end{tikzpicture}
$$
(furthermore, the single edge incident to $u$ or to $v$ is labeled by $1$).
As a result, each edge $\{x,y\}$ gets port labels that depend solely on the
colors of $x$ and~$y$:
\begin{equation}
  \label{eq:port-label-edge}
\begin{tikzpicture}[scale=.85]
\path (0,0) node[vertex] (a) {}
      (1,0) node[vvertex] (b)  {} edge (a)
       (b) node[draw=none,fill=none,above left] {\scriptsize 1}
       (a) node[draw=none,fill=none,above right] {\scriptsize 1};
\end{tikzpicture}\,,\quad
\begin{tikzpicture}[scale=.85]
\path (0,0) node[vvertex] (a) {}
      (1,0) node[vvertex] (b)  {} edge (a)
       (b) node[draw=none,fill=none,above left] {\scriptsize 2}
       (a) node[draw=none,fill=none,above right] {\scriptsize 2};
\end{tikzpicture}\,,\quad
\begin{tikzpicture}[scale=.85]
\path (0,0) node[vertex] (a) {}
      (1,0) node[avertex] (b)  {} edge (a)
       (b) node[draw=none,fill=none,above left] {\scriptsize 1}
       (a) node[draw=none,fill=none,above right] {\scriptsize 2};
\end{tikzpicture}\,,\quad
\begin{tikzpicture}[scale=.85]
\path (0,0) node[avertex] (a) {}
      (1,0) node[bvertex] (b)  {} edge (a)
       (b) node[draw=none,fill=none,above left] {\scriptsize 1}
       (a) node[draw=none,fill=none,above right] {\scriptsize 2};
\end{tikzpicture}\,,\quad
\begin{tikzpicture}[scale=.85]
\path (0,0) node[avertex] (a) {}
      (1,0) node[bbvertex] (b)  {} edge (a)
       (b) node[draw=none,fill=none,above left] {\scriptsize 1}
       (a) node[draw=none,fill=none,above right] {\scriptsize 3};
\end{tikzpicture}\,,\quad
\begin{tikzpicture}[scale=.85]
\path (0,0) node[bvertex] (a) {}
      (1,0) node[bbvertex] (b)  {} edge (a)
       (b) node[draw=none,fill=none,above left] {\scriptsize 2}
       (a) node[draw=none,fill=none,above right] {\scriptsize 2};
\end{tikzpicture}\,.
\end{equation}
Denote the obtained port-numbered vertex-colored graphs by $G'$ and $H'$
and their uncolored versions by $G'_0$ and $H'_0$.

Endow the universal covers $U=U_u(G)$ and $W=U_v(H)$ with port-numbering
according to \refeq{port-label-edge}. Denote the resulting port-numbered trees
by $U'$ and $W'$ and their uncolored versions by $U'_0$ and $W'_0$.
Note that $U'$ and $W'$ are the universal covers of $G'$ and $H'$,
and $U'_0$ and $W'_0$ are the universal covers of $G'_0$ and $H'_0$ respectively.
Any isomorphism between $U$ and $W$ is an isomorphism between $U'$ and $W'$
and, hence, also between $U'_0$ and $W'_0$. The same holds true as well for
the truncations at any depth. By \refeq{iso-r}, this implies that 
$U_u^r(G'_0)\cong U_v^r(H'_0)$ for $r=2n-O(\sqrt n)$.
Finally, $U_u(G'_0)\not\cong U_v(H'_0)$. This readily follows from~\refeq{non-iso}.  
\end{remark}

Using Lemma \ref{lem:UvsC}, from Theorem \ref{thm:nancy} 
we derive the following fact.

\begin{corollary}
For each $n$, there are $n$-vertex graphs $G$ and $H$
with disjoint vertex sets such that $\stabi{G\cup H}=(2-o(1))n$.
\end{corollary}

\section{Quantifier depth in two-variable counting logic}\label{s:logic}

We consider first-order logic \fo{} for graphs with 
two binary relations for adjacency
and equality of vertices. Let \fo2 denote the fragment of
\fo{} consisting of formulas built from only two variables.
We will enrich the language by using expressions of the type
$\exists^mx\Psi(x)$ in order to say that there are at least $m$ vertices $x$ with property 
$\Psi(x)$. The \emph{counting quantifier} $\exists^m$ contributes 1 in the quantifier depth 
irrespectively of the value of $m$. The corresponding syntactic extension of \fo2
will be denoted by~\twoclogic.
The following fact follows directly from the definition \refeq{Ci} by induction on~$i$.

\begin{lemma}\label{lem:def-color}
  For any possible $C^i$-color $c$ there is a formula $\Phi(x)$ in \twoclogic
of quantifier depth $i$ such that, for every graph $G$ and its vertex $u$,
$C^i(u)=c$ if and only if $\Phi(x)$ is true on $G$ for $x=u$.\noproof
\end{lemma}

We say that a formula $\Phi$ \emph{distinguishes} $G$ and $H$
if it is true on exactly one of the graphs.
Let $D(G,H)$ denote the minimum quantifier depth of such a formula in \twoclogic.
The 2-pebble counting game \cite{ImmermanL90} on graphs $G$ and $H$
differs from its bisimulation version described in the preceding section in that
the sets $A$ and $B$ in each round are constrained only by the equality $|B|=|A|$.
Instead, Duplicator must now ensure that the pebbling after each round
determines a partial isomorphism, that is, 
the two pebbled vertices in $G$ are adjacent/equal
exactly if the pebbled vertices in $H$ are adjacent/equal
(in the bisimulation version, they were automatically not equal and adjacent in both graphs).

\begin{lemma}[Immerman and Lander \cite{ImmermanL90}]\label{lem:IL}
  $D(G,H)\le r$ if and only if Spoiler has a winning strategy 
in the $r$-round counting game on $G$ and~$H$.\noproof
\end{lemma}

\begin{lemma}\label{lem:indistC}\mbox{}
Suppose that at least one of graphs $G$ and $H$ is connected.
Let $s=\stabi G$. Then $G$ and $H$ are indistinguishable in \twoclogic
if and only if these graphs have equal number of vertices and
\begin{equation}
  \label{eq:CsGH}
\setdef{C^{s+1}(u)}{u\in V(G)}=\setdef{C^{s+1}(v)}{v\in V(H)}. 
\end{equation}
\end{lemma}

Note that Equality \refeq{CsGH} does not require that each color
occurs in $G$ and $H$ the same number of times.

\begin{proof}
If $G$ and $H$ have different number of vertices, these graphs
are distinguished by a sentence of quantifier depth 1 in \twoclogic.
If Equality \refeq{CsGH} is not true, $G$ and $H$ can be distinguished 
by a \twoclogic-sentence of quantifier depth $s+2$ by Lemma \ref{lem:def-color}. 

Conversely, assume that $G$ and $H$ have equal number of vertices 
and Equality \refeq{CsGH} is true. Then
$\setdef{C^{s}(u)}{u\in V(G)}=\setdef{C^{s}(v)}{v\in V(H)}$
because differently colored vertices cannot get equal colors later.
Since $s=\stabi G$, for every $c,c'\in\setdef{C^{s}(u)}{u\in V(G)}$
any vertex $u\in V(G)$ with $C^{s}(u)=c$ has the same number $m_{c,c'}$ of
neighbors $w$ such that $C^{s}(w)=c'$. The matrix $M=(m_{c,c'})$
is called the \emph{degree refinement matrix} of $G$.
Note that the same property, with the same matrix $M$, holds true also for $H$
for else \refeq{CsGH} would be false 
(this implies, in particular, that $\stabi H\le s$).

Now, let us show that
\begin{equation}
  \label{eq:msets-equal}
\msetdef{C^{s}(u)}{u\in V(G)}=\msetdef{C^{s}(v)}{v\in V(H)},
\end{equation}
which means that $G$ and $H$ contain the same number of vertices
of each $C^s$-color $c$.
Regarded as an adjacency matrix, the matrix $M$ determines
directed multigraphs on the vertex sets 
$\setdef{C^{s}(u)}{u\in V(G)}$ and $\setdef{C^{s}(v)}{v\in V(H)}$.
Since at least one of $G$ and $H$ is connected, the directed
multigraphs determined by $M$ are connected too.
Assume for a while that there is a color $c$ such that
$|\setdef{u\in V(G)}{C^{s}(u)=c}|<|\setdef{v\in V(H)}{C^{s}(v)=c}|$.
Note that the inequality 
$
|\setdef{u\in V(G)}{C^{s}(u)=c'}|<|\setdef{v\in V(H)}{C^{s}(v)=c'}|
$
is true also for any adjacent color $c'$ (such that $m_{c,c'}\ne0$).
It follows by connectedness that $|V(G)|<|V(H)|$.
This contradiction proves~\refeq{msets-equal}.

Using \refeq{msets-equal}, we can show that Duplicator has a winning strategy
in the counting game on $G$ and $H$ for any number of rounds.
Assume that the pebble $q$ occupies the vertices $u$ in $G$ and $v$ in $H$
such that $C^s(u)=C^s(v)$; note that the condition \refeq{msets-equal}
allows Duplicator to ensure this in the first round. 
Assume that Spoiler now plays with the other pebble $p$
and specifies a set $A$. Duplicator is able to respond
with a set $B$ such that, for each $C^s$-color $c$, the number of
vertices in $B$ colored in $c$ is the same as the number of
such vertices in $A$. Moreover, based on the fact that $G$ and $H$
have the same degree refinement matrix $M$, Duplicator can ensure
that the same is true even if only the neighbors or only the non-neighbors 
of $u$ and $v$ are considered. This allows Duplicator not to lose this
round and also to ensure that the vertices under $p$ in $G$ and $H$
will have the same $C^s$-color. 
\end{proof}

\begin{lemma}\label{lem:DvsStab}
  If $G$ and $H$ are distinguishable in \twoclogic,
then $D(G,H)\le\stabi G+2$.
\end{lemma}

\begin{proof}
  Note that $D(G,H)=D(\barG,\barH)$, where $\barG$ and $\barH$
denote the complements of $G$ and $H$ respectively.
Considering the complements if necessary, we therefore can
suppose that $G$ is connected. Thus, we are in the conditions of Lemma \ref{lem:indistC}.
Let $s=\stabi G$. Lemma \ref{lem:indistC} implies that $G$ and $H$ can be
distinguished by specifying a $C^{s+1}$-color $c$ occurring in only one of these graphs.
By Lemma \ref{lem:def-color}, the existence of a vertex colored in $s$
can be expressed by a statement of quantifier depth~$s+2$.
\end{proof}

\begin{theorem}\label{Phard}
  The problem of deciding if two given connected graphs have a common cover 
is \p-complete under \ac0-reductions.
\end{theorem}

\begin{proof}
  The problem is in \p due to Angluin's algorithm; see Section \ref{ss:angluin}.
The hardness for \p follows by reduction from the \twoclogic-equivalence problem,
which is to decide whether two given graphs $G$ and $H$ are distinguishable
or equivalent in \twoclogic. 
Note that this problem is equivalent to its restriction to connected graphs.
Lemmas \ref{lem:common} and~\ref{lem:indistC}
readily imply that connected $G$ and $H$ with the same number of vertices are \twoclogic-equivalent
if and only if they have a common covering graph. The
\p-completeness of the \twoclogic-equivalence problem is established by Grohe~\cite{Grohe99}.\footnote{%
The \p-completeness of the \twoclogic-equivalence is stated in \cite{Grohe99} for directed graphs with
noting that the proof works also for (undirected) vertex-colored graphs.
The \twoclogic-equivalence problem for colored graphs easily reduces to
the version of this problem for uncolored graphs. Given a colored graph $G$
with $n$ vertices, attach $n+i$ new vertices to each vertex of $G$ colored in the $i$-th color,
remove all colors, and denote the resulting graph by $G'$. Then
$G$ and $H$ are equivalent in \twoclogic if and only if $G'$ and $H'$ are equivalent in this logic.
This easily follows from the characterization of the \twoclogic-equivalence by the
bijection game suggested by Hella~\cite{Hella96}.}
\end{proof}

The following result shows that the \twoclogic-equivalence type of a graph with $n$ vertices
is definable with quantifier depth at most $n+1$ and that this bound is
asymptotically tight.
\hide{
The equivalence class of a graph $G$ in \twoclogic consists of all
graphs indistinguishable from $G$ in this logic.
There is a sentence $\Phi$ in \twoclogic that defines this class  
in the sense that $\Phi$ is true on $G$ (hence an all \twoclogic-equivalent graphs)
but false on every non-\twoclogic-equivalent graph. Let $D(G)$ denote the
minimum quantifier depth of such a sentence.
Define the function $D(n)$ as the maximum $D(G)$
over $n$-vertex graphs. Theorem \ref{thm:qdepth} readily implies that $D(n)=(1-o(1))n$.
}

\begin{theorem}\label{thm:qdepth}\mbox{}
  \begin{enumerate}
  \item[\bf 1.] 
If $n$-vertex graphs $G$ and $H$ are distinguishable in \twoclogic,
then $D(G,H)\le n+1$.
\item[\bf 2.]
For each $n$,
there are $n$-vertex graphs $G$ and $H$ distinguishable in \twoclogic such
that $D(G,H) > n-8\,\sqrt n$.
  \end{enumerate}
\end{theorem}

The proof of Theorem \ref{thm:qdepth} takes the rest of this section.
Part 1 follows immediately from Lemma \ref{lem:DvsStab}.
To prove Part 2, we consider the same graphs $G_{s,t}$ and $H_{s,t}$ 
as in Section \ref{s:Norris}. By Lemma \ref{lem:IL}, 
it suffices to show that Spoiler has a winning strategy in the
2-pebble counting game on these graphs but Duplicator is able
to resist longer than $n-8\,\sqrt n$ rounds. 
The former fact (Spoiler wins) follows from the proof of Lemma \ref{lem:GstHst}.1;
however, we give a much shorter and simpler argument for it.
More effort is needed to prove the latter fact (Duplicator resists for long);
the game analysis is now harder because Spoiler has more freedom than
in the bisimulation version. 

We first design a winning strategy for Spoiler in 
the 2-pebble counting game on $G_{s,t}$ and $H_{s,t}$,
showing that these graphs are distinguishable in \twoclogic.
Let us use the notation and the notions introduced in Section \ref{s:Norris}
for analysis of the bisimulation version of this game.

In the first round, let Spoiler pebble a vertex $a_1$ in $G_{s,t}$
at the maximum distance from the vertex $u$. That is, $\ell(a_1)=l+5$
where $l$ is defined by \refeq{ldef}. Whatever Duplicator's response $b_1$ in $H_{s,t}$ is,
$\ell(b_1)\le l+2$. In the subsequent rounds, Spoiler pebbles, one by one,
adjacent vertices along a shortest path from $b_1$ to $v$.
When Spoiler reaches the vertex $v$ of degree 1, Duplicator arrives at a vertex of degree at least 2
in $G_{s,t}$. Spoiler wins in the next round.

In order to prove the lower bound for $D(G,H)$, we design a strategy for Duplicator
allowing her to stay alive for a long time.
This strategy consists in ensuring that
\begin{itemize}
\item 
$a_1$ and $b_1$ are of the same type and, furthermore,
$a_i$ and $b_i$ are of the same type as long as
\begin{equation}
  \label{eq:>0'}
\min\{\ell(a_{i-1}),\ell(b_{i-1})\}>0\text{ or }\ell(a_{i-1})=\ell(b_{i-1})=0.
\end{equation}
\end{itemize}
The condition \refeq{>0'} ensures that, if $a_{i-1}$ and $b_{i-1}$ are of the same type,
then they have equally many neighbors and non-neighbors of each type
and, hence, Duplicator not only is able not to lose the next $i$-th round
but even to secure $a_i$ and $b_i$ are of the same type.
To keep \refeq{>0'} true as long as possible, Duplicator tries in each round to fulfil
at least one of the following three conditions
\begin{eqnarray}
\ell(a_i)=\ell(b_i)\quad\ \ \mbox{}&&\label{eq:ll}\\
\min\{\ell(a_i),\ell(b_i)\}&\ge&l-1\label{eq:lll}\\
\min\{\ell(a_i),\ell(b_i)\}&\ge&\min\{\ell(a_{i-1}),\ell(b_{i-1})\}-1\label{eq:llll}
\end{eqnarray}
Let us use an inductive argument to show that such a strategy does exist.

\begin{lemma}\mbox{}\label{lem:strategy}
  \begin{enumerate}
  \item[\bf 1.] 
In the first round of the Immerman-Lander game on $G_{s,t}$ and $H_{s,t}$,
Duplicator is able to ensure pebbling vertices $a_1$ and $b_1$ of the same type
satisfying at least one of the conditions \refeq{ll} and \refeq{lll} for $i=1$.
\item[\bf 2.] 
If $a_{i-1}$ and $b_{i-1}$ are of the same type and $\ell(a_{i-1})=\ell(b_{i-1})$,
then in the next round Duplicator is able to ensure pebbling vertices $a_i$ and $b_i$ 
of the same type satisfying at least one of the conditions \refeq{ll} and~\refeq{lll}.
\item[\bf 3.] 
If $a_{i-1}$ and $b_{i-1}$ are of the same type, $\ell(a_{i-1})\ne\ell(b_{i-1})$,
and $\min\{\ell(a_{i-1}),\ell(b_{i-1})\}\allowbreak>1$,
then in the next round Duplicator is able to ensure pebbling vertices $a_i$ and $b_i$ 
of the same type satisfying at least one of the three conditions \refeq{ll}--\refeq{llll}.
  \end{enumerate}
\end{lemma}

\begin{proof}
  {\bf 1.}
Let $\phi\function{V(G_{s,t})}{V(H_{s,t})}$ be a bijection that preserves
vertex types and is a partial isomorphism from $G_{s,t}$ to $H_{s,t}$
up to level $l-1$. Note that, for any $a\in V(G_{s,t})$,
\begin{equation}
  \label{eq:ella}
\ell(a)=\ell(\phi(a))\text{ or }\min\{\ell(a),\ell(\phi(a))\}\ge l.
\end{equation}
Duplicator mirrors Spoiler's move according to $\phi$,
that is, after Spoiler specifies a set $A$, Duplicator
responds with $B=\phi(A)$ if $A\subseteq V(G_{s,t})$ or $B=\phi^{-1}(A)$ if $A\subseteq V(H_{s,t})$.
After Spoiler pebbles $b\in B$, Duplicator pebbles the vertex $a\in A$ such that
$a=\phi^{-1}(b)$ or $a=\phi(b)$ respectively.

\smallskip

{\bf 2.}
If $a_{i-1}$ and $b_{i-1}$ are pebbled in the preceding round,
Duplicator has to modify $\phi$ to $\phi'$ so that
\begin{equation}
  \label{eq:phiab}
\phi'(a_{i-1})=b_{i-1}.  
\end{equation}
Another condition to obey is
\begin{equation}
  \label{eq:phiNab}
\phi'(N(a_{i-1}))=N(b_{i-1}),
\end{equation}
which is possible because the assumption $\ell(a_{i-1})=\ell(b_{i-1})$
implies that $a_{i-1}$ and $b_{i-1}$ are of the same degree.
By Property A of the construction of $G_{s,t}$ and $H_{s,t}$, 
the map $\phi'$ can be defined on $N(a_{i-1})$ so that it
respects the vertex types.
The equality \refeq{phiNab} ensures the condition
$$
a_i\in N(a_{i-1})\iff b_i\in N(b_{i-1}),
$$
which is necessary for Duplicator's survival starting from the second round.
If $\ell(a_{i-1})=\ell(b_{i-1})\le l-1$, then $\phi'$ can still be chosen
to be a partial isomorphism up to level $l-1$, and this case is much similar
to Part 1. It is also possible that $\ell(a_{i-1})=\ell(b_{i-1})$ is equal to $l$ or to $l+1$,
for example, if $a_{i-1}=u'$ and $b_{i-1}=v'_j$. Then $\phi'$ can be supposed
to be a partial isomorphism up to level $l-2$ and then
$$
\ell(a)=\ell(\phi'(a))\text{ or }\min\{\ell(a),\ell(\phi'(a))\}\ge l-1,
$$
ensuring \refeq{ll} or~\refeq{lll}.

\smallskip

{\bf 3.}
If $\ell(a_{i-1})\ne\ell(b_{i-1})$, the modification of $\phi$ to $\phi'$ has to be described
with more care.
First of all, $\phi$ is modified on $\{a_{i-1}\}\cup N(a_{i-1})$
so that \refeq{phiab} and \refeq{phiNab} are fulfilled.
Suppose that, as a result, $\phi'(a)=b$ for $a\in \{a_{i-1}\}\cup N(a_{i-1})$ and 
$b\in \{b_{i-1}\}\cup N(b_{i-1})$. We are now forced to define the new image of 
$a'=\phi^{-1}(b)$ and the new preimage of $b'=\phi(a)$, and we do this by setting
$\phi'(a')=b'$. Note that the vertex types are preserved.

The above definition of $\phi'$ is ambiguous only when
$a_{i-1}$ and $b_{i-1}$ are of type \vtype{bvertex} and $|\ell(a_{i-1})-\ell(b_{i-1})|=2$.
If in this case $a'\in N(a_{i-1})$, then $\phi'(a')$ is defined from the very beginning and
should not be redefined any more. 
We remove this collision by supposing that $\phi(N(a_{i-1}))=N(b_{i-1})$
(if necessary, apply an automorphism of $G_{s,t}$ transposing two pairs of \vtype{bvertex}-vertices
in the intermediate level between $\ell(a_{i-1})$ and $\ell(b_{i-1})$). Then 
$\phi'$ coincides with $\phi$ on $N(a_{i-1})$, and there is no need to modify $\phi$ further.

Assume that the $i$-th round has been played.
If $\phi'(a_i)=\phi(a_i)$, then $a=a_i$ satisfies \refeq{ella}, which implies \refeq{ll}
or \refeq{lll}. If $\phi'(a_i)\ne\phi(a_i)$, then we have \refeq{llll} because $\phi'$
differs from $\phi$ only on the levels neighboring with $\ell(a_{i-1})$ and $\ell(b_{i-1})$.
\end{proof}

To complete the proof of Theorem \ref{thm:qdepth}, fix a strategy for Duplicator
as in Lemma \ref{lem:strategy} and an arbitrary winning strategy for Spoiler.
Assume that Spoiler wins in the $(r+1)$-th round. Note that 
$$
\ell(a_r)\ne\ell(b_r)\text{ and }\min\{\ell(a_r),\ell(b_r)\}=0
$$ 
because otherwise we would get a contradiction with Lemma \ref{lem:strategy}.
Let $k$ be the smallest number such that $\ell(a_{i})\ne\ell(b_{i})$
for all $k\le i\le r$. Note that either $k=1$ or $\ell(a_{k-1})=\ell(b_{k-1})$.
Parts 1 and 2 of Lemma \ref{lem:strategy} imply that $\min\{\ell(a_k),\ell(b_k)\}\ge l-1$.
Consider the largest index $m\ge k$ for which $\min\{\ell(a_m),\ell(b_m)\}\ge l-1$.
By Lemma \ref{lem:strategy}.3, Spoiler needs no less than $l-1$ rounds to decrease
$\min\{\ell(a_i),\ell(b_i)\}$ from $l-1$ to $0$. Thus, $r\ge m+l-1\ge l$.
Like in the proof of Theorem \ref{thm:nancy}, we take $G=G_{2t+1,t}$ and $H=H_{2t+1,t}$, 
adding new dummy vertices if necessary, and conclude that $D(G,H)\ge l> n-8\,\sqrt n$.

\section{Comments and questions}\label{s:open}
\mbox{}

\que
Let $G$ and $H$ be connected $n$-vertex graphs with diameters at most $D$.
Boldi and Vigna \cite{BoldiV02} notice that the isomorphism of the universal covers
$U_x(G)$ and $U_y(H)$ is implied
by the isomorphism of their truncations $U_x^{n+D}(G)$ and $U_y^{n+D}(H)$. 
Since $D\le n-1$, the upper bound of $n+D$ for the distinguishing truncation depth is 
more advantageous than $2n-1$.
Our result shows that the bound of $n+D$ is, in general, also asymptotically tight.
If $D=(1-o(1))n$, then a lower bound of $(1-o(1))(n+D)=(2-o(1))n$
is given directly by Theorem \ref{thm:nancy}.
 If $D=o(n)$ (in particular, if $D=O(1)$),
we can obtain a lower bound of $(1-o(1))(n+D)=(1-o(1))n$
by slightly modifying our construction of graphs $G_{s,t}$ and $H_{s,t}$.
In order to decrease the diameter, we appropriately choose
a set $T$ of vertex types, add a new vertex to each of the graphs
and connect it to all vertices whose type is in~$T$.

Is the upper bound of $n+D$ tight if $D=c\, n$ for a constant $c\in(0,1)$? For example,
if $D=\frac12\,n$, is then the upper bound of $\frac32\, n$ tight up to $o(n)$?

\que
Define $S(n)$ to be the maximum $\stabi G$
over all graphs $G$ with $n$ vertices. Combining
Theorem \ref{thm:nancy} and Lemma \ref{lem:UvsC}, we conclude that
$$
n-8\sqrt2\sqrt n<S(n)<n.
$$ 
Is it true that $S(n)=n-O(1)$? Alternatively,
does there exists a function $f(n)$ going to
the infinity as $n$ increases such that $S(n)<n-f(n)$?

\que
Similarly to \twoclogic, let \kclogic denote the $k$-variable logic
with counting quantifiers.
The equivalence class of a graph $G$ in \kclogic consists of all
graphs indistinguishable from $G$ in this logic.
There is a sentence $\Phi$ in \kclogic that defines this class  
in the sense that $\Phi$ is true on $G$ (hence an all \kclogic-equivalent graphs)
but false on every non-\kclogic-equivalent graph. Let $D^k(G)$ denote the
minimum quantifier depth of such a sentence.
Define the function $D^k(n)$ as the maximum $D^k(G)$
over $n$-vertex graphs. Theorem \ref{thm:qdepth} readily implies that $D^2(n)=(1-o(1))n$.

If $k\ge3$, it is known \cite{DawarLW95,PikhurkoV11} that $D^k(n)\le n^{k-1}$ for $n\ge2$.
How tight is this bound? A linear lower bound is shown by Fürer~\cite{Fuerer01}.

\end{document}